\documentclass[11pt,letterpaper]{article}
\usepackage{authblk}

\usepackage[letterpaper,margin=1in]{geometry}
\usepackage{microtype}
\usepackage{amsmath,amssymb,amsthm,mathtools}
\usepackage{aliascnt}
\usepackage{algorithm}
\usepackage[noend]{algpseudocode}
\usepackage{needspace}
\makeatletter
\renewcommand{\theALG@line}{\thealgorithm.\arabic{ALG@line}}
\providecommand{\theHALG@line}{}
\renewcommand{\theHALG@line}{\thealgorithm.\arabic{ALG@line}}
\makeatother
\usepackage{enumitem}
\usepackage{hyperref}
\usepackage[nameinlink,noabbrev,capitalise]{cleveref}
\usepackage[numbers,square,sort&compress]{natbib}
\usepackage{url}

\crefname{theorem}{Theorem}{Theorems}
\Crefname{theorem}{Theorem}{Theorems}
\crefname{lemma}{Lemma}{Lemmas}
\Crefname{lemma}{Lemma}{Lemmas}
\crefname{corollary}{Corollary}{Corollaries}
\Crefname{corollary}{Corollary}{Corollaries}
\crefname{algorithm}{Algorithm}{Algorithms}
\Crefname{algorithm}{Algorithm}{Algorithms}
\crefname{section}{Section}{Sections}
\Crefname{section}{Section}{Sections}

\hypersetup{
  hidelinks,
  pdftitle={One Discrete Gaussian Sample in 2 to the n over 2 plus o of n Time},
  pdfauthor={},
  pdfkeywords={discrete Gaussian sampling, shortest vector problem, lattice sparsification, smoothing parameter}
}

\newtheorem{theorem}{Theorem}[section]
\newaliascnt{lemma}{theorem}
\newtheorem{lemma}[lemma]{Lemma}
\aliascntresetthe{lemma}
\newaliascnt{corollary}{theorem}
\newtheorem{corollary}[corollary]{Corollary}
\aliascntresetthe{corollary}
\theoremstyle{definition}
\newaliascnt{definition}{theorem}
\newtheorem{definition}[definition]{Definition}
\aliascntresetthe{definition}

\newcommand{\R}{\mathbb{R}}
\newcommand{\Q}{\mathbb{Q}}
\newcommand{\Z}{\mathbb{Z}}
\newcommand{\eps}{\varepsilon}
\newcommand{\poly}{\operatorname{poly}}
\newcommand{\polylog}{\operatorname{polylog}}
\newcommand{\dist}{\operatorname{dist}}
\newcommand{\CVP}{\mathsf{CVP}}
\newcommand{\ip}[2]{\langle #1,#2\rangle}
\newcommand{\norm}[1]{\lVert #1\rVert_2}
\newcommand{\ckl}{c_{\mathrm{Lev}}}
\newcommand{\betakl}{\beta_{\mathrm{Lev}}}

\setlist[itemize]{leftmargin=1.6em,itemsep=0.2em,topsep=0.3em}
\setlist[enumerate]{leftmargin=1.8em,itemsep=0.2em,topsep=0.3em}

\title{One Discrete Gaussian Sample in $2^{n/2+o(n)}$ Time} 
\author{Jiseung Kim}
\affil{Jeonbuk National University \\ jiseungkim@jbnu.ac.kr}
\date{}

\begin{document}
\maketitle 

\begin{abstract}
Aggarwal, Dadush, Regev, and Stephens-Davidowitz (ADRS; STOC 2015) sample $2^{n/2}$ discrete Gaussians
at an arbitrary parameter in $2^{n+o(n)}$ time, and above smoothing in $2^{n/2+o(n)}$ time.
They ask whether the latter bound suffices for one sample at an arbitrary parameter.
We answer this question affirmatively: for every rank-$n$ lattice $L\subseteq\R^n$ specified by a
rational basis and every rational $s^2>0$, we
produce one sample from $D_{L,s}$ within statistical distance $\exp(-\Omega(n^3))$ in expected
$2^{n/2+o(n)}$ time and $2^{n/2+o(n)}$ space on every execution.  The algorithm samples from random superlattices
that are smooth at the required scale with constant probability and outputs the first point in $L$;
a Gaussian-mass comparison shows
that the $2^{n/2}$ samples produced by one ADRS call contain a point of $L$ with
inverse-polynomial probability.  The factor $2^{n/2}$ is tight in this Gaussian-mass comparison.
For every fixed rational $\alpha<1.4697$, the same comparison gives a sub-$2^n$
algorithm for exact CVP on targets satisfying
$\dist(y,L)\le\alpha\lambda_1(L)$, without a uniqueness assumption, and an exact-SVP algorithm in
$2^{0.7315n+o(n)}$ time.
\end{abstract}

\section{Introduction}
For a lattice $L\subseteq\R^n$ and a parameter $s>0$, the discrete Gaussian $D_{L,s}$ assigns
each $x\in L$ probability proportional to $\exp(-\pi\norm{x}^2/s^2)$.  Discrete Gaussian
sampling (DGS) is a basic algorithmic primitive in the geometry of numbers and in lattice
cryptography~\citep{MR07,GPV08}.  At large parameters it is efficiently samplable; at the small
parameters relevant to exact lattice problems, its complexity changes sharply.

Aggarwal, Dadush, Regev, and Stephens-Davidowitz (ADRS) give an arbitrary-parameter algorithm
that outputs $2^{n/2}$ samples in $2^{n+o(n)}$ time and space.  Above the smoothing parameter,
their second algorithm outputs the same number of samples in $2^{n/2+o(n)}$ time and space.  They explicitly
ask whether this roughly $2^{n/2}$ running time can be achieved at arbitrary parameters, at least
when only one sample is requested~\citep{ADRS15}.

Our main result shows that the faster above-smoothing running time already suffices for one sample at
every parameter.

\begin{theorem}[Main result, informal; see \Cref{thm:one-dgs}]
Given a rational basis of a rank-$n$ lattice $L\subseteq\R^n$ and a rational squared parameter
$s^2>0$, a random vector whose distribution is $\exp(-\Omega(n^3))$-close to $D_{L,s}$ can be
produced in expected $2^{n/2+o(n)}$ time.  Every execution uses at most
$2^{n/2+o(n)}$ space.
\end{theorem}

Here ``arbitrary'' refers to the scale regime: $s$ may lie above or below the smoothing parameter
and need not be related to $\lambda_1(L)$.  The requirement $s^2\in\Q$ is only the finite input
convention for the Turing-machine statement.  The algorithm needs neither a reduced basis nor an
estimate of $\lambda_1(L)$.  \Cref{sec:precision,thm:one-dgs} give the polynomial dependence on
the input length and the randomized Turing-machine implementation.  The algorithm still stores
one output list of the ADRS sampler.  The improvement is in the time needed for one sample at an
arbitrary parameter, not in space.

\paragraph{Applications to exact lattice problems.}
A one-dimensional augmentation turns the mass comparison into an algorithm for exact CVP.  For
every fixed rational $\alpha<1.4697$, it runs in sub-$2^n$
time on inputs satisfying $\dist(y,L)\le\alpha\lambda_1(L)$, without assuming a unique closest
vector.  At $\alpha=1$, the exponent is approximately $0.7315$; \Cref{thm:cvp-alpha} gives the full
dependence on $\alpha$.

The same pointwise sampling bound gives a secondary exact-SVP consequence.  At the largest scale
where the Gaussian mass remains subexponential, it yields time $2^{0.7315n+o(n)}$ and space
$2^{n/2+o(n)}$.  \Cref{thm:svp-above-smoothing} separates this exponent into the cost of one
above-smoothing call and the number of calls needed to see a shortest vector.

\subsection{Related work}

The closest starting point is ADRS: they give the first $2^{n+o(n)}$-time sampler at arbitrary
parameters and a $2^{n/2+o(n)}$-time sampler above smoothing~\citep{ADRS15}.  Our main result asks
for only one output and removes the smoothing restriction at the latter cost.  Gaussian measures
and smoothing also underlie lattice cryptography~\citep{MR07,GPV08}, and the smoothing parameter
itself can be approximated algorithmically~\citep{CDLP13}.  Dimension-preserving reductions in the
opposite direction use exact CVP for shifted DGS and exact SVP for centered DGS~\citep{Ste16}.
Polynomial-time sampling under lattice extensions and restrictions, and quantum DGS, are studied
respectively in \citep{EWY23,CSS26}.

For exact SVP, DGS is already a central algorithmic tool.  The main classical line runs from the
Voronoi-cell algorithm~\citep{MV13} through the ADRS sampler~\citep{ADRS15} and the simpler
pair-and-average algorithm~\citep{AS18}.  Above-smoothing DGS and BDD also yield time--space and
classical--quantum tradeoffs~\citep{CCL18,ACKS25}; faster bounds are known for random
lattices~\citep{PoulyShen26}.  Our secondary exact-SVP result instead uses lattice
sparsification~\citep{DadushKun13}: the superlattice is chosen from its dual Gaussian mass, and the
two lattices are compared at different scales.

For exact Euclidean CVP, the relevant distinction is between unrestricted targets and targets close
to the lattice.  Micciancio and Voulgaris gave a deterministic $\widetilde O(4^n)$-time,
$\widetilde O(2^n)$-space algorithm~\citep{MV13}; Aggarwal, Dadush, and Stephens-Davidowitz obtained
$2^{n+o(n)}$ time and space using shifted discrete Gaussians~\citep{ADS15}.  Aggarwal and
Stephens-Davidowitz later simplified the sampling step while retaining the exact-CVP
recursion~\citep{AS18}.  In the close-target regime, ADRS solve $(0.422-o(1))$-BDD in
$2^{n/2+o(n)}$ time, while Dadush, Regev, and Stephens-Davidowitz study $\alpha$-BDDP and reductions
for approximate CVP~\citep{ADRS15,DRS14}.  Under SETH, exact $\CVP_p$ has no
$2^{(1-\varepsilon)n}$-time algorithm for finite $p\in[1,\infty)\setminus2\mathbb Z$; for $p=2$,
ABGS instead prove a barrier for a specified class of natural reductions~\citep{ABGS21}.  Abboud
and Kumar obtain a sub-$2^n$ algorithm for $(0,1)$-CVP under a coordinate-size condition~\citep{AK25}.
Our CVP result covers every rational lattice and target satisfying the stated distance
bound.

\paragraph{Concurrent work.}
Gao, Feng, and Hu independently obtain an exact-SVP algorithm in time
$2^{0.7314n+o(n)}$ and space $2^{n/2+o(n)}$ by sampling above smoothing on a
prime-index superlattice and keeping the samples in the input lattice~\citep{GFH26}.
Our exact-SVP consequence is similar, with a slightly larger exponent of approximately $0.7315$.
The two analyses choose the superlattice differently.  Gao, Feng, and Hu tune a prime to a guessed
shortest-vector scale and separate the dual mass in $pL^*$ from the remaining mass.  We use the
kernel of a homomorphism from $L^*$ to $\Z/2^m\Z$, choosing $m$ from the dual Gaussian mass at the
given parameter.  Our main theorem treats arbitrary-parameter one-sample DGS and also proves
tightness of the Gaussian-mass comparison used in its analysis.

Hhan gives a different exact-SVP algorithm, using the Hessian of the periodic Gaussian at the
midpoint of a shortest vector to recover that vector.  The algorithm runs in expected time
$2^{0.603867n+o(n)}$ and space $2^{n/2+o(n)}$~\citep[Theorem~1.1]{Hhan26}.  Combined
with the dimension-preserving reduction of Stephens-Davidowitz~\citep[Theorem~4.6]{Ste16}, this
yields arbitrary-parameter centered DGS with the same exponential running time and
inverse-polynomial statistical error.  For DGS, its exponent is larger than $1/2$, and the
reduction gives weaker statistical error than $\exp(-\Omega(n^3))$.  Our algorithm obtains one
sample directly from the Gaussian-mass comparison whose $2^{n/2}$ dependence is tight.

\subsection{Technical overview}

The proof is organized around one quantitative question: can a superlattice that is smooth enough
for the fast ADRS sampler still contain points of $L$ often enough?  The same dual Gaussian mass controls
both the smoothness of a random superlattice and the loss from restricting its samples back to $L$.

\paragraph{Rejection and Gaussian mass.}
Suppose that $\Lambda\supseteq L$ is a superlattice from which we can sample at parameter $s$.
For an ideal sample $X\sim D_{\Lambda,s}$,
\[
 \Pr[X\in L]=\frac{\rho_s(L)}{\rho_s(\Lambda)},
 \qquad
 \Pr[X=x\mid X\in L]=\frac{e^{-\pi\norm{x}^2/s^2}}{\rho_s(L)}.
\]
Thus rejection from $\Lambda$ has the desired conditional distribution, and its cost is determined
by $\rho_s(\Lambda)/\rho_s(L)$.  We must choose $\Lambda$ smooth enough for the fast ADRS sampler
while keeping
\[
 \frac{\rho_s(\Lambda)}{\rho_s(L)}\le2^{n/2}\poly(n),
\]
so that the $2^{n/2}$ samples returned by one call include a point of $L$ with
inverse-polynomial probability.

The smooth ADRS sampler cannot in general be run directly on $L$ when $s$ is below smoothing.
The arbitrary-parameter sampler reaches this range by repeatedly transforming a large list, and
asking for only one final sample does not remove its $2^{n+o(n)}$ cost.  We instead set
$t=s/\sqrt2$, construct a random superlattice that is smooth at $t$ with constant probability, and
run the ADRS sampler on it at
$s=\sqrt2t$.

\paragraph{Constructing a smooth superlattice.}
Smoothness is controlled by short vectors in the dual.  Write
$\Theta_t:=\rho_{1/t}(L^*)$.  Following the lattice-sparsification method of Dadush and
Kun~\citep{DadushKun13}, choose the kernel $M_z\subseteq L^*$ of a random homomorphism to
$\Z/q\Z$, where $q=2^m$, and put $L_z=M_z^*\supseteq L$.  A nonzero dual vector satisfies the
random linear condition with probability roughly $1/q$; vectors divisible by powers of two require the
more careful estimate in \Cref{lem:inclusion}.  When $q$ is within a polynomial factor of
$\Theta_t$, the expected nonzero Gaussian mass remaining in $M_z$ is a constant.  Consequently,
with constant probability,
\[
 \rho_{1/t}(M_z\setminus\{0\})\le\frac14,
\]
and $L_z$ is smooth at parameter $t$.

The value $\Theta_t$ need not be computed.  From the rational input $(B,t^2)$ we obtain a
polynomial-size range containing the relevant exponent $m$, try every value in that range, and use
a fresh random kernel each time.  This remains efficient even if $[L_z:L]=2^m$ is large: a
rational basis of $L_z$ needs only $O(m)$ additional bits, and membership in $L$ is tested in the
original coordinates.

\paragraph{The mass comparison.}
The same dual mass that determines the index also controls the rejection probability.  For the
distinguished value $m_*(L,t)$, Poisson summation gives
\begin{equation}
\label{eq:intro-mass-comparison}
 \rho_{\sqrt2t}(L_z)
 \le 40(2m-1)\,2^{n/2}\rho_t(L).
\end{equation}
To see where the exponent comes from, Poisson summation writes the left-hand side as the product
of the index, $(\sqrt2t)^n/\det L$, and a bounded dual mass.  The choice of index is governed by
\[
 \Theta_t=\rho_{1/t}(L^*)=\frac{\det L}{t^n}\rho_t(L).
\]
After substitution, the determinant and the power of $t$ cancel.  The only exponential loss is
$(\sqrt2)^n=2^{n/2}$; the factor $2m-1$ accounts for the possible powers of two dividing a dual
coefficient vector.  Since $\rho_s(L)\ge\rho_t(L)$, \eqref{eq:intro-mass-comparison}
implies that one sample from $D_{L_z,s}$ lies in $L$ with probability at least
$1/(40(2m-1)2^{n/2})$.

The factor $2^{n/2}$ in this comparison is tight.  Every superlattice that is smooth at $t$ has
Gaussian mass at parameter $\sqrt2t$ at least $\frac23 2^{n/2}\rho_t(L)$.  This statement concerns
the comparison with $\rho_t(L)$; it is not a lower bound on the acceptance probability or running
time of every rejection sampler.

\paragraph{Recovering one sample.}
An ADRS call on $L_z$ at parameter $s$ returns $2^{n/2}$ samples.  This list size cancels the
$2^{-n/2}$ membership probability, so the list contains a point of $L$ with inverse-polynomial
probability.  We scan the list and return its first point in $L$.  Conditional on success, the
displayed identity at the start of the overview gives exactly $D_{L,s}$.

The unknown modulus and the approximate ADRS output require a final wrapper.  The algorithm does
not know which value of $m$ gives the mass comparison, and a call may return a variable-length list
or $\bot$.  We try
all admissible values of $m$ and repeat the resulting finite experiment until it succeeds.
\Cref{lem:first-success} shows that the first accepted point is still distributed as $D_{L,s}$, while
\Cref{sec:precision} implements the experiment with fair random bits and total statistical error
$\exp(-\Omega(n^3))$.

\paragraph{Shifted targets.}
The construction also applies after adjoining a target coordinate, but an additional mass ratio
appears.  Given a target $y$ and a scale $u>0$, form the one-dimensional extension
\[
 \widetilde\Lambda
   =\{(x-ky,ku/\sqrt2):x\in L_z,\ k\in\Z\}.
\]
Conditioning a sample from $D_{\widetilde\Lambda,\sqrt2u}$ on $k=1$ and $x\in L$ gives exactly
$D_{L-y,\sqrt2u}$.  Before conditioning, its probability is exactly
\[
 e^{-\pi/4}\frac{\rho_{\sqrt2u}(L-y)}{\rho_{\sqrt2u}(\widetilde\Lambda)}.
\]
The mass comparison lower-bounds this probability by the ratio
$\rho_{\sqrt2u}(L-y)/\rho_u(L)$ times $2^{-n/2}$ and a polynomial factor.  This ratio can be
arbitrarily small, so the construction does not
give arbitrary-parameter shifted DGS in $2^{n/2+o(n)}$ time.  For the CVP application, however, it
suffices to lower-bound the weight of one closest vector.  The distance guarantee and the Gaussian-mass
bound in \Cref{lem:adrs-mass} make that probability explicit; optimizing $u$ yields
\Cref{thm:cvp-alpha}.

\section{Preliminaries}

A full-rank lattice in $\R^n$ is written
$L=L(B):=B\Z^n$, where the columns of the nonsingular matrix $B$ form a basis.  We use
$\ip xy$ and $\norm x$ for the Euclidean inner product and norm, write
$\det L:=|\det B|$, and identify the dual lattice with
$L^*=B^{-T}\Z^n$.  For a discrete set $A\subseteq\R^n$ and $s>0$, write
\[
\rho_s(A)=\sum_{x\in A}e^{-\pi\norm{x}^2/s^2},
\qquad
D_{L,s}(x)=\frac{e^{-\pi\norm{x}^2/s^2}}{\rho_s(L)}\quad(x\in L).
\]
Thus $D_{L,s}$ assigns probability proportional to $e^{-\pi\norm{x}^2/s^2}$ to each $x\in L$.
The smoothing parameter is
$\eta_\eps(L)=\inf\{s>0:\rho_{1/s}(L^*\setminus\{0\})\le\eps\}$.
We write $\lambda_i(L)$ for the $i$th successive minimum; in particular, $\lambda_1(L)$ is the
length of a shortest nonzero vector.  For a point $y\in\R^n$,
$\dist(y,L):=\min_{x\in L}\norm{y-x}$.  We write $d_{\mathrm{TV}}(P,Q)$ for total variation
distance and use $\bot$ for an aborted sampler call.  For $a\in\R$, put
$[a]_+:=\max\{0,a\}$.
If $L\subseteq\Lambda$ are full-rank lattices, then
$[\Lambda:L]=\det L/\det\Lambda$ denotes their index.
For a rank-$r$ lattice, put $N_r:=\lceil2^{r/2}\rceil$.
We repeatedly use the Poisson summation identity
\begin{equation}
\label{eq:poisson}
\rho_s(L)=\frac{s^n}{\det L}\rho_{1/s}(L^*).
\end{equation}

\paragraph{Representation and model.}
The input consists of a nonsingular rational basis $B$ and a positive rational squared scale
$\sigma=s^2$; the CVP application additionally receives a rational target $y$.  For a rational
object $u$, write $\operatorname{bits}(u)$ for the length of its standard binary encoding, and set
$\beta_{\mathrm{in}}:=\operatorname{bits}(B)$,
$\beta_s:=\operatorname{bits}(\sigma)$, and
$\beta_y:=\operatorname{bits}(y)$.  Running time counts bit operations, and space is the maximum
number of work-tape bits used on any execution.  Lattice membership and norm comparisons use exact
rational arithmetic; \Cref{sec:precision} implements the real-valued sampling probabilities by
certified intervals.

\begin{lemma}
\label{lem:dilation}
For every lattice $\Lambda$, integer $k\ge1$, and $s>0$,
\[
\rho_s(k\Lambda\setminus\{0\})\le\frac1k\rho_s(\Lambda\setminus\{0\}).
\]
\end{lemma}
\begin{proof}
On each primitive ray, the Gaussian weights form a nonincreasing sequence $(a_j)_{j\ge1}$ and
$k a_{kj}\le\sum_{\ell=(j-1)k+1}^{jk}a_\ell$.  Sum over $j$, both signs, and all primitive rays.
\end{proof}

\begin{lemma}[{\citealp[Lemma~2.3]{ADRS15}}]
\label{lem:gaussian-scaling}
For every lattice $\Lambda$, $s>0$, and $\gamma\ge1$,
\[
\rho_{\gamma s}(\Lambda)\le\gamma^n\rho_s(\Lambda).
\]
\end{lemma}
\begin{lemma}
\label{lem:shifted-mass}
For every lattice $\Lambda\subseteq\R^r$, shift $y\in\R^r$, and $s>0$,
$\rho_s(\Lambda-y)\le\rho_s(\Lambda)$.
\end{lemma}
\begin{proof}
Poisson summation writes $\rho_s(\Lambda-y)$ as
$s^r(\det\Lambda)^{-1}\sum_{w\in\Lambda^*}e^{-\pi s^2\norm w^2}e^{2\pi i\ip wy}$.
The triangle inequality bounds its absolute value by the same sum with every phase removed, which
is $\rho_s(\Lambda)$.
\end{proof}
\begin{lemma}[{\citealp[Lemma~2.13]{DRS14}}]
\label{lem:dgs-tail}
Let $X\sim D_{L,s}$ and $T\ge1/\sqrt{2\pi}$.  Then
\[
\Pr[\norm X>Ts\sqrt n]
<
\left(\sqrt{2\pi e}\,T e^{-\pi T^2}\right)^n.
\]
In particular, taking $T=n$ makes the probability $\exp(-\Omega(n^3))$.
\end{lemma}
Below smoothing, the ADRS sampler may return a shorter list.  The following definition records the
distributional guarantee that remains valid.

\begin{definition}[{\citealp[Definition~5.1]{ADRS15}}]
\label{def:variable-dgs}
Let $\nu\ge0$, let $\varphi$ map lattices to nonnegative reals, and let $N\in\Z_{>0}$.  On input
$(L,s)$, an algorithm satisfies this guarantee if its output is within statistical distance
$\nu$ of the following experiment: first choose $M\in\Z_{\ge0}$ independently of the sample
values, and then output $M$ independent samples from $D_{L,s}$.  If $s>\varphi(L)$, then $M=N$.
\end{definition}

The independence of $M$ is essential below.  It allows us to take the subsequence lying in a
sublattice even when the call is made below the threshold.

\begin{theorem}[{\citealp[Theorem~5.11]{ADRS15}}]
\label{thm:adrs-smooth-sampler}
Let $N=N_n$ and $\kappa=\Omega(n)$.  There is an algorithm satisfying
\Cref{def:variable-dgs} on rank-$n$ lattices with statistical distance $\exp(-\Omega(\kappa))$ and threshold
$\varphi(L)=\sqrt2\,\eta_{1/2}(L)$.  In the ADRS arithmetic-cost model, its running time and
space are $2^{n/2+\polylog(\kappa)+o(n)}$ up to polynomial factors.
\end{theorem}

For rational inputs, \Cref{thm:adrs-rational} gives a randomized Turing-machine implementation with
the same exponential rate, polynomial dependence on the input length, and the stated space bound on
every execution.

We also use the standard hybrid bound for adaptive calls.

\begin{lemma}
\label{lem:adaptive-tv}
Consider a randomized computation making at most $K$ oracle calls.  The input to a call may depend
on the preceding transcript.  If, for every transcript prefix, the actual distribution of call
$i$ is within total variation distance $\nu_i$ of its ideal distribution, then the actual and ideal
complete transcripts are within $\sum_{i=1}^K\nu_i$.
\end{lemma}
\begin{proof}
For $0\le j\le K$, let the first $j$ calls be ideal and the remaining calls actual.  Two adjacent
hybrids differ only at call $j$.  Conditioned on their common prefix, this difference is at most
$\nu_j$; applying the same randomized continuation cannot increase total variation distance.
Average over prefixes and sum the adjacent distances.
\end{proof}

\section{Sampling below the smoothing parameter}

Given the requested sampling scale $s$, put $t=s/\sqrt2$.  The fast ADRS sampler cannot be run on
$L$ when $s\le\sqrt2\eta_{1/2}(L)$.  We therefore construct a random superlattice that is smooth at
$t$ with constant probability
and bound the Gaussian-mass loss incurred by restricting its samples back to $L$.

\subsection{Dual sparsification and the superlattice \texorpdfstring{$L_z$}{Lz}}

Let $D=B^{-T}$ be the dual basis, fix $t>0$ with rational square
$\tau=t^2$, and put $\Theta_t:=\rho_{1/t}(L^*)$.  For $L_z=M_z^*$ to be smooth at $t$, it suffices
to make the nonzero mass of $M_z\subseteq L^*$ at parameter $1/t$ small.  We obtain $M_z$ by
imposing one random linear condition modulo $q$ on the dual coefficients; powers of two provide an explicit sequence
of candidate indices.

Fix $q=2^m$.  Choose uniformly a primitive vector $z\in(\Z/q\Z)^n$, meaning that at least one
coordinate is odd, and define
\[
K_z=\{a\in\Z^n:\ip za\equiv0\pmod q\},
\qquad
M_z=DK_z\subseteq L^*,
\qquad
L_z=M_z^*\supseteq L.
\]

\begin{lemma}
\label{lem:kernel-basis}
For primitive $z$, the three indices satisfy $[\Z^n:K_z]=[L^*:M_z]=[L_z:L]=q$, and a basis of
$L_z$ is computable in $\poly(n,\beta_{\mathrm{in}},m)$ bit operations.
\end{lemma}
\begin{proof}
The map $a\mapsto\ip za\bmod q$ is surjective because one coordinate of $z$ is a unit modulo
$q$.  After permuting coordinates, assume $z_1$ is odd and let $u=z_1^{-1}\bmod q$.  The columns
\[
qe_1,
\qquad
e_j-(uz_j\bmod q)e_1\quad(2\le j\le n)
\]
form a basis matrix $C_z$ of $K_z$ with $|\det C_z|=q$.  Hence
$L_z=BC_z^{-T}\Z^n$.
\end{proof}

\begin{lemma}
\label{lem:inclusion}
Let $a\in\Z^n\setminus q\Z^n$, and let $2^j$, $0\le j<m$, be the largest common power of two
dividing its coordinates.  Then $\Pr_z[a\in K_z]\le2^{j+1}/q$.  Vectors in $q\Z^n$ are included
with probability one.
\end{lemma}
\begin{proof}
For uniform unrestricted $z$, the homomorphism $z\mapsto\ip za$ has image size $q/2^j$, so its
zero fraction is $2^j/q$.  Conditioning on primitive $z$ increases probabilities by less than
$2$ for $n\ge2$.  For $n=1$, a primitive $z$ is a unit modulo $q$, and the probability is zero.
\end{proof}

Let $m_*=m_*(L,t)$ be the least positive integer satisfying
\begin{equation}
\label{eq:mstar}
2^{m_*}\ge16(2m_*+1)\Theta_t,
\end{equation}
and put $q_*=2^{m_*}$.

\begin{theorem}
\label{thm:random-smoothing}
For $q=q_*$,
\[
\Pr_z\!\left[\rho_{1/t}(M_z\setminus\{0\})\le\frac14\right]\ge\frac34.
\]
Consequently, on this event, $\eta_{1/2}(L_z)<t$.
\end{theorem}
\begin{proof}
Partition $\Z^n\setminus\{0\}$ according to the largest common power of two dividing their
coordinates.  Vectors divisible by $q$ have
dual Gaussian mass at most $\Theta_t/q$ by \Cref{lem:dilation}.  For $0\le j<m$, vectors divisible
by $2^j$ have mass at most $\Theta_t/2^j$, while \Cref{lem:inclusion} bounds their inclusion
probability by $2^{j+1}/q$.  Thus
\[
\mathbb E_z\rho_{1/t}(M_z\setminus\{0\})
\le
\frac{\Theta_t}{q}+
\sum_{j=0}^{m-1}\frac{2^{j+1}}q\frac{\Theta_t}{2^j}
=
\frac{(2m+1)\Theta_t}{q}
\le\frac1{16}.
\]
Markov's inequality gives the claimed probability.  The dual nonzero mass at parameter $1/t$
is then strictly below $1/2$, so continuity gives $\eta_{1/2}(L_z)<t$.
\end{proof}

The algorithm does not compute $\Theta_t$.  Exact rational linear algebra gives a rational lower
bound $\underline\lambda_B\le\sigma_{\min}(B^{-T})$.  Certified square-root arithmetic applied to
$\tau$ gives a positive rational $\underline u\le t\underline\lambda_B$ of polynomial bit length, and
hence $\Theta_t\le\rho_{1/\underline u}(\Z^n)\le(1+1/\underline u)^n$.
Define the explicit rational upper bound and its bit-length proxy by
\[
\overline\Theta_t:=\left(1+\frac1{\underline u}\right)^n,
\qquad
\ell_t:=\max\left\{0,\left\lceil\log_2\overline\Theta_t\right\rceil\right\},
\]
and set
\begin{equation}
\label{eq:explicit-J}
J_{\mathrm{mod}}(B,\tau):=\ell_t+2\left\lceil\log_2(\ell_t+2)\right\rceil+8.
\end{equation}
The two elementary bounds
\[
2^{J_{\mathrm{mod}}(B,\tau)}\ge256\overline\Theta_t(\ell_t+2)^2
\quad\text{and}\quad
2J_{\mathrm{mod}}(B,\tau)+1\le16(\ell_t+2)^2
\]
show that
$2^{J_{\mathrm{mod}}(B,\tau)}\ge16(2J_{\mathrm{mod}}(B,\tau)+1)\overline\Theta_t$.
Hence $m_*(L,t)\le J_{\mathrm{mod}}(B,\tau)$.  The bound is computable in polynomial time from
$(B,\tau)$ and is polynomial in the input length.  Although $q=2^m$ may be large, its binary
representation has only $O(m)$ bits; the algorithm never enumerates the quotient.

\subsection{Comparing the Gaussian masses of \texorpdfstring{$L_z$ and $L$}{Lz and L}}

Smoothness permits sampling from $L_z$, but it does not by itself ensure that a sample lies in
$L$.  The required probability is controlled by comparing the mass of $L_z$ at the sampling scale
$\sqrt2t$ with the mass of $L$ at parameter $t$.

Minimality in \eqref{eq:mstar} implies $m_*>1$ and
\begin{equation}
\label{eq:q-upper}
q_*=2^{m_*}<32(2m_*-1)\Theta_t.
\end{equation}

\begin{theorem}
\label{thm:mass-comparison}
On the good event of \Cref{thm:random-smoothing},
\[
\rho_{\sqrt2t}(L_z)
\le
40(2m_*-1)\,2^{n/2}\rho_t(L).
\]
\end{theorem}
\begin{proof}
Since $M_z=L_z^*\subseteq L^*$ has index $q_*$,
$\det L_z=\det L/q_*$.  Applying \eqref{eq:poisson} to $L_z$ at $\sqrt2t$ gives
$\rho_{\sqrt2t}(L_z)=q_*(\sqrt2t)^n\rho_{1/(\sqrt2t)}(M_z)/\det L$.
On the good event, $\rho_{1/(\sqrt2t)}(M_z)\le\rho_{1/t}(M_z)\le5/4$.  By
\eqref{eq:q-upper} and a second application of \eqref{eq:poisson},
$\Theta_t=\rho_{1/t}(L^*)=(\det L/t^n)\rho_t(L)$.
Substitution cancels the determinant and scale factors:
\[
\begin{aligned}
\rho_{\sqrt2t}(L_z)
&\le
\frac54\cdot32(2m_*-1)
\rho_{1/t}(L^*)
\frac{(\sqrt2t)^n}{\det L}\\
&=
40(2m_*-1)2^{n/2}\rho_t(L).
\end{aligned}
\]
\end{proof}

\subsection{Sampling from the original lattice}

At the distinguished modulus $m_*(L,t)$, \Cref{thm:random-smoothing} makes $L_z$ smooth at $t$, while
\Cref{thm:mass-comparison} bounds the loss from restricting the samples back to $L$.  We sample on
$L_z$ and output precisely the points that belong to $L$.

The analysis uses $K_z$ and $M_z$, but the algorithm only needs a basis of $L_z$.  The construction
in \Cref{lem:kernel-basis} is explicit: choose an odd coordinate of $z$, form a basis matrix $C_z$
of $K_z$, and use $BC_z^{-T}$ as a basis of $L_z$.

\begin{algorithm}[H]
\caption{Sampling from $D_{L_z,\sqrt2t}$ and returning the vectors in $L$}
\label{alg:superlattice-run}
\begin{algorithmic}[1]
\Require Rational basis $B$ of $L$, $\tau=t^2\in\Q_{>0}$, positive integers $m$ and $\kappa$
\State $q\gets2^m$; sample $z$ uniformly from the primitive vectors in $(\Z/q\Z)^n$
\State choose $i$ with $z_i$ odd and set $u\gets z_i^{-1}\bmod q$
\State let $C_z$ have columns $qe_i$ and
$e_j-(uz_j\bmod q)e_i$ for $j\ne i$; set $G_z\gets BC_z^{-T}$
\State call the ADRS sampler on $G_z\Z^n$ with squared parameter $2\tau$ and statistical parameter $\kappa$
\If{the sampler returns $\bot$} \State \Return the empty sequence \EndIf
\State \Return the subsequence of returned vectors $x$ satisfying $B^{-1}x\in\Z^n$
\end{algorithmic}
\end{algorithm}

The algorithm never enumerates $L_z/L$.  The basis $G_z$ has
$\poly(n,\beta_{\mathrm{in}},m)$ bits,
and membership in $L$ is an exact integrality test in the coordinates of $B$.  The output may be
processed sequentially and discarded before the next call.  Only the vectors from one ADRS call
are stored.

Fix $L$, $t>0$, and the superlattice $L_z$.  On the good event,
$\eta_{1/2}(L_z)<t$, so \Cref{thm:adrs-smooth-sampler} at parameter $s=\sqrt2t$ returns a list of
$N:=N_n$ samples jointly close to $D_{L_z,\sqrt2t}^{\,N}$.

\begin{theorem}
\label{thm:setwise-list}
Let $S\subseteq L$ be finite.  For a choice of $z$ satisfying \Cref{thm:random-smoothing}, a list of
$N$ independent samples from $D_{L_z,\sqrt2t}$ intersects $S$ with probability at least
\[
\frac1{80(2m_*-1)}
\min\left\{
1,
\frac{\rho_{\sqrt2t}(S)}{\rho_t(L)}
\right\}.
\]
For a random choice of $z$, the joint probability over $z$ and the samples satisfies the same lower
bound with $80$ replaced by $160$.
\end{theorem}
\begin{proof}
For one ideal sample, $p_z(S)=\rho_{\sqrt2t}(S)/\rho_{\sqrt2t}(L_z)$.  By
\Cref{thm:mass-comparison},
\[
Np_z(S)\ge\frac{\rho_{\sqrt2t}(S)}{40(2m_*-1)\rho_t(L)}.
\]
For $0\le p\le1$ and integer $N\ge1$,
$1-(1-p)^N\ge\frac12\min\{1,Np\}$.  This proves the first bound.  Multiplying by the
probability $3/4$ from \Cref{thm:random-smoothing} and weakening the absolute constant gives the second.
\end{proof}

\begin{corollary}
\label{cor:pointwise-list}
For every fixed $x\in L$ and every choice of $z$ satisfying \Cref{thm:random-smoothing}, a list of
$N$ independent samples from $D_{L_z,\sqrt2t}$ contains $x$ with probability at least
\[
\frac1{80(2m_*-1)}
\min\left\{
1,
\frac{e^{-\pi\norm{x}^2/(2t^2)}}{\rho_t(L)}
\right\}.
\]
The corresponding run with random $z$ satisfies the same statement with the constant
$80$ replaced by $160$.
\end{corollary}
\begin{proof}
Apply \Cref{thm:setwise-list} with $S=\{x\}$.
\end{proof}

These statements concern whether an entire output list intersects a set; they do not condition on
that event.  Replacing the ideal list by the ADRS output therefore changes its probability by at
most the total variation distance.

\subsection{Tightness of the Gaussian-mass comparison}

The upper bound in \Cref{thm:mass-comparison} compares $\rho_{\sqrt2t}(L_z)$ with $\rho_t(L)$.
The factor $2^{n/2}$ in this comparison is unavoidable for every superlattice that is smooth at
$t$.

\begin{theorem}
\label{thm:index-lower}
Let $L\subseteq\Lambda$ be full-rank lattices and suppose $\eta_{1/2}(\Lambda)\le t$.  Then
\[
[\Lambda:L]\ge\frac23\rho_{1/t}(L^*).
\]
\end{theorem}
\begin{proof}
Put $M=\Lambda^*\subseteq L^*$ and $q=[\Lambda:L]$.  Continuity and monotonicity of the
dual mass show that $\eta_{1/2}(\Lambda)\le t$ implies
$\rho_{1/t}(M\setminus\{0\})\le1/2$.  Poisson summation gives
$\rho_t(\Lambda)/\rho_t(L)=q\rho_{1/t}(M)/\rho_{1/t}(L^*)$.
The left-hand side is at least one, while smoothing gives $\rho_{1/t}(M)\le3/2$.  Rearranging proves
the claim.
\end{proof}

\begin{theorem}
\label{thm:mass-floor}
Under the hypotheses of \Cref{thm:index-lower},
\[
\rho_{\sqrt2t}(\Lambda)\ge\frac23\,2^{n/2}\rho_t(L).
\]
\end{theorem}
\begin{proof}
The identity \eqref{eq:poisson} at $\sqrt2t$ and
$\rho_{1/(\sqrt2t)}(\Lambda^*)\ge1$ give
$\rho_{\sqrt2t}(\Lambda)\ge[\Lambda:L](\sqrt2t)^n/\det L$.
By \Cref{thm:index-lower} and another application of \eqref{eq:poisson},
$[\Lambda:L]\ge\frac23\det L\,\rho_t(L)/t^n$.  Substitution proves the claim.
\end{proof}

Thus \Cref{thm:mass-comparison} is tight in its $2^{n/2}$ dependence when its denominator is
$\rho_t(L)$.

\section{One sample at an arbitrary parameter}

The distinguished modulus $m_*(L,t)$ is not known, so the algorithm tries all admissible values of $m$.  At an
incorrect modulus the ADRS call may return fewer samples, but its list length is independent of
their values.  This independence lets the algorithm return the first point in $L$ without biasing
its distribution.

\begin{lemma}
\label{lem:variable-restriction}
Let $L\subseteq\Lambda$ be full-rank lattices and $s>0$.  Let $M\ge0$ be independent of
$X_1,X_2,\ldots$, where the $X_i$ are independent samples from $D_{\Lambda,s}$.  Return the first
$X_i\in L$ with $i\le M$, and return $\bot$ if no such index exists.  Conditional on returning a
vector, the output distribution is $D_{L,s}$.
\end{lemma}
\begin{proof}
Fix $M$ and the index $i$ of the first sample in $L$.  Independence leaves $X_i$ distributed as
$D_{\Lambda,s}$ conditioned on $X_i\in L$, and for every $x\in L$,
\[
\Pr[X_i=x\mid X_i\in L]
=
\frac{e^{-\pi\norm{x}^2/s^2}}{\rho_s(L)}.
\]
This distribution is independent of $M$ and $i$, so averaging over them gives the claim.
\end{proof}

\begin{lemma}
\label{lem:first-success}
Let $Y_1,\ldots,Y_T$ be independent random variables in $\Omega\cup\{\bot\}$ such that
$\mathcal L(Y_i\mid Y_i\ne\bot)=D$ whenever $Y_i$ succeeds with positive probability.  Conditional
on at least one success, the first non-$\bot$ value has distribution $D$.  If an actual joint
distribution is within total variation distance $\nu$ of this ideal sequence and the ideal success
probability is $p>\nu$, then its first accepted output, conditioned on success, is within
$2\nu/(p-\nu)$ of $D$.
\end{lemma}
\begin{proof}
Condition on the index of the first success for the exact statement.  The approximate statement is
the conditioning inequality
$d_{\mathrm{TV}}(P(\cdot\mid A),Q(\cdot\mid A))\le2\nu/(Q(A)-\nu)$.
\end{proof}

Choose once and for all a positive integer $C_\kappa$ large enough to dominate the absolute
constants in the finite collection of error bounds in \Cref{sec:precision}.  This is a fixed,
hardwired constant of the algorithm rather than part of the input.

\begin{algorithm}[H]
\caption{Sampling one vector from $D_{L,s}$}
\label{alg:one-dgs}
\begin{algorithmic}[1]
\Require Rational basis $B$ of $L$ and squared parameter $\sigma=s^2\in\Q_{>0}$
\State $\beta_{\mathrm{in}}\gets\operatorname{bits}(B)$ and
$\beta_s\gets\operatorname{bits}(\sigma)$
\State $\tau\gets\sigma/2$ and compute $J:=J_{\mathrm{mod}}(B,\tau)$ from \eqref{eq:explicit-J}
\State $\kappa_0\gets\left\lceil C_\kappa\bigl(n^3+\log(2+\beta_{\mathrm{in}}+\beta_s+J)\bigr)\right\rceil$
\State $T\gets\lceil640J\kappa_0\rceil$
\While{true}
  \For{$m=1,\ldots,J$}
    \For{$r=1,\ldots,T$}
      \State $W\gets$ the vectors returned by \Cref{alg:superlattice-run} on
      $(B,\tau,m,\kappa_0)$ using fresh independent randomness
      \If{$W\ne\emptyset$} \State \Return the first vector of $W$ \EndIf
    \EndFor
  \EndFor
\EndWhile
\end{algorithmic}
\end{algorithm}

\begin{theorem}[Arbitrary-parameter one-sample DGS]
\label{thm:one-dgs}
Let $B\in\Q^{n\times n}$ be nonsingular, let $L=L(B)$, and let
$\sigma=s^2\in\Q_{>0}$.  \Cref{alg:one-dgs} outputs a random vector whose distribution is
$\exp(-\Omega(n^3))$-close to $D_{L,s}$.  Its expected running time is
\[
2^{n/2+o(n)}\poly(n,\beta_{\mathrm{in}},\beta_s),
\]
bit operations, and every execution uses at most
$2^{n/2+o(n)}\poly(n,\beta_{\mathrm{in}},\beta_s)$ bits of space.
\end{theorem}
\begin{proof}
Put $t=\sqrt{\tau}=s/\sqrt2$.  For every modulus and choice of $z$, the ideal ADRS list satisfies
\Cref{lem:variable-restriction}; hence a successful run returns $D_{L,s}$.  This conditional distribution is
independent of the modulus, $z$, and list length, so \Cref{lem:first-success} gives the same distribution
for the first successful ideal run.

For the distinguished modulus $m_*$ and a choice of $z$ satisfying \Cref{thm:random-smoothing},
\Cref{thm:mass-comparison} gives
$\rho_s(L)/\rho_s(L_z)\ge1/(40(2m_*-1)2^{n/2})$,
where $\rho_{s/\sqrt2}(L)\le\rho_s(L)$.  The $N=N_n$ samples therefore intersect
$L$ with probability at least $1/(80(2m_*-1))$.  The required event holds with probability at least
$3/4$, and $m_*\le J$ by \eqref{eq:explicit-J}.  Thus, after accounting for the random choice of
$z$, each run at this modulus succeeds with probability $\Omega(1/J)$.  The
$T=\lceil640J\kappa_0\rceil$ independent runs therefore make one iteration fail with probability
$\exp(-\Omega(\kappa_0))$.
The expected number of iterations is $1+o(1)$.

One iteration makes $JT=O(J^2\kappa_0)$ sampler calls.  Each call costs
$2^{n/2+o(n)}\poly(n,\beta_{\mathrm{in}},\beta_s)$ bit operations by
\Cref{thm:adrs-smooth-sampler}, and the superlattice basis and membership tests have polynomial bit
complexity.  Rejection sampling a primitive $z$ succeeds with probability at least $1/2$, so its
expected cost is polynomial as well.  The returned vectors are processed sequentially and discarded before the next call,
so only the vectors returned by one call are stored at a time.
This proves the expected-time bound and the claimed bound on space.

The total variation errors of the polynomially many calls in one iteration sum to
$\nu=\exp(-\Omega(\kappa_0))$.  Since the ideal iteration succeeds with probability
$p=1-\exp(-\Omega(\kappa_0))$, \Cref{lem:first-success} bounds the conditional output error by
$2\nu/(p-\nu)=\exp(-\Omega(\kappa_0))$.  Independent restarts preserve this conditional distribution
rather than accumulating error over an unbounded number of iterations.  The exact implementation
and the space bound on every execution are proved in \Cref{sec:precision}.
\end{proof}

\section{Implementation on a randomized Turing machine}
\label{sec:precision}

It remains to implement the ADRS random choices and real-number comparisons on rational inputs.
The implementation below preserves the expected running time, bounds space on every execution,
and loses only negligible statistical distance.  Stopping after a constant multiple of the
expected running time would not suffice: it could discard a constant fraction of the output distribution.

The difficulty is not the computability of the required probabilities.  Exact comparisons may
request arbitrarily many random bits, and discrete Gaussian samples have unbounded support.
\Cref{lem:truncate-sampling} isolates the truncation argument.  The next three lemmas verify its
comparison hypotheses for the rational Gram computations and the one-dimensional sampler, while
\Cref{lem:coefficient-tail,lem:coefficient-tail-gram} bound the stored representations.
\Cref{thm:adrs-rational} then combines these facts for one ADRS call.

\subsection{Exact implementation of the random choices}

Every Gaussian parameter is represented by a positive rational squared scale $\sigma=s^2$.
For rational lattice vectors, the quotient $\norm{x}^2/\sigma$ is rational, and quantities such as
$e^{-\pi\norm{x}^2/\sigma}$ are evaluated by certified intervals.

\begin{definition}
A real quantity $x$ determined by the rational input is \emph{efficiently computable} if, on input
$q$, an algorithm returns dyadic
rationals $\underline x_q\le x\le\overline x_q$ with $\overline x_q-\underline x_q\le2^{-q}$ using a
number of bit operations polynomial in $q$, $n$, and the input bit length.
\end{definition}

Rational Gram--Schmidt data, square roots of positive rationals,
exponentials, and the one-dimensional theta masses used by the
Brakerski--Langlois--Peikert--Regev--Stehl\'e (BLPRS) sampler are efficiently
computable~\citep[Section~5.1]{BLPRS13}.  Narrow theta sums are truncated in the primal domain; wide sums are evaluated after
Poisson transformation.  Gaussian-integral tails certify the remainder in both cases.

\begin{lemma}
\label{lem:lazy-bernoulli}
If $p\in[0,1]$ is efficiently computable, an exact $\operatorname{Bernoulli}(p)$ variable can be
generated by a procedure that reads finitely many bits with probability one.  If the number of
bits read is limited to $Q$, the procedure aborts with probability at most $6\cdot2^{-Q}$.
\end{lemma}

\begin{proof}
Generate a lazy uniform binary expansion $U=0.U_1U_2\ldots$.  After $q$ bits, compare the
dyadic interval containing $U$ with a width-$2^{-q}$ interval containing $p$.  Decide when the
intervals are disjoint.  Failure by level $q$ implies that $U$ lies in an interval of length at most
$6\cdot2^{-q}$ around $p$.
\end{proof}

\Needspace{12\baselineskip}
\begin{lemma}
\label{lem:truncate-sampling}
Let $\mathcal A(x)$ be an exact Las Vegas procedure with discrete output and expected bit cost at
most a computable integer $T(x)\ge1$.  Assume the following.
\begin{enumerate}[label=(\roman*),leftmargin=2.3em]
\item Outside an event of probability $\delta_{\mathrm{repr}}$, its stored exact data use at most
$S(x)$ bits.
\item Conditional on every reachable transcript prefix before this representation cutoff, each
comparison is driven by one or two independent unread uniform tails and efficiently computable
certified intervals.
\item There is a computable $q_0(x)=\poly(|x|)$ such that the conditional probability that a
comparison remains unresolved after $q_0(x)+q$ further refinement levels is at most
$6\cdot2^{-q}$; one level reads at most one bit from each active tail.
\item At most $b(x)=\poly(|x|)$ tails are stored simultaneously, and each is used in at most
$\poly(|x|)$ comparisons before it is discarded.
\end{enumerate}

For $K\ge1$, set
\[
H=\left\lceil2^KT(x)\right\rceil,
\qquad
Q=q_0(x)+\left\lceil\log_2(6H)\right\rceil+K.
\]
Abort after $H$ bit operations, when any lazy comparison requests more than $Q$ further bits from
one of its active tails, or when the stored exact data use more than $S(x)$ bits.  Then the
truncated procedure has the following guarantees.
\begin{enumerate}[label=(\alph*),leftmargin=2.3em]
\item Its output is within $\delta_{\mathrm{repr}}+2^{1-K}$ in total variation distance of
$\mathcal A(x)$.
\item Its expected running time is at most $T(x)$.
\item It uses at most
\[
S(x)+b(x)\poly(|x|,Q,\log H)
\]
bits.
\end{enumerate}
\end{lemma}
\begin{proof}
Markov's inequality bounds the probability that the procedure runs for more than $H$ bit operations
by $2^{-K}$.  Before that event there are at
most $H$ comparisons requiring adaptive precision.  The conditional
probability that any one of them requests more than $Q$ further bits from an active tail is at most
$6\cdot2^{-(Q-q_0(x))}$.  This bound holds
after conditioning on an arbitrary reachable prefix, so a union bound gives at most
$6H2^{-(Q-q_0(x))}\le2^{-K}$ even when later thresholds depend on earlier outcomes.  Adding the
event that the stored exact data use more than $S(x)$ bits proves the statistical bound.  The procedure stores
only the stated number of random prefixes, counters of length $O(\log H)$, and certified threshold
intervals of precision $Q$.
Stopping can only decrease the expected running time.
\end{proof}

\begin{lemma}
\label{lem:lazy-comparison}
Fix a reachable transcript prefix and suppose that every discrete envelope value, scale, and
coefficient used by a comparison has at most $B$ bits.  Conditional on this prefix, let one operand
$X$ have a density bounded by $d_X$ and be independent of the other operand.  Suppose that the
comparison refines certified intervals $I_X(\ell)$ and $I_Y(\ell)$ so that, after $\ell$ further
refinement rounds,
\[
 d_X\bigl(|I_X(\ell)|+|I_Y(\ell)|\bigr)\le 2^{p(B)-\ell}
\]
for a fixed polynomial $p$.  There is a computable $q_0=\poly(B)$ such that the probability that
the comparison remains unresolved after $q_0+q$ rounds is at most $6\cdot2^{-q}$.
\end{lemma}
\begin{proof}
Condition on the complete value of the other operand.  If the certified intervals overlap, $X$
lies in an interval of length at most $2(|I_X(\ell)|+|I_Y(\ell)|)$ around that value.  The
conditional probability is therefore at most $2^{p(B)+1-\ell}$.  Taking
$\ell=q_0+q$ with $q_0\ge p(B)+1$ proves the claim.  The fixed-threshold case is identical.
\end{proof}

\begin{lemma}
\label{lem:gram-computation}
Let $\Gamma\in\Q^{r\times r}$ be positive definite.  Every lattice operation used by the ADRS
sampler can be implemented directly in integer coefficient coordinates in time polynomial in $r$
and the bit length of $\Gamma$.  In particular, the implementation does not require a rational
Euclidean basis for the lattice represented by $\Gamma$.
\end{lemma}
\begin{proof}
Exact inner products are $a^T\Gamma b$, the dual quadratic form is represented by $\Gamma^{-1}$,
and Gram--Schmidt quantities are efficiently computable from $\Gamma$.
All lattice transformations in the sampler are rational changes of coefficient coordinates with
denominators of polynomial bit length.  Their
Gram matrices are obtained by congruence transformations $U^T\Gamma U$.  Exact membership and equality
are integer-coordinate questions, while all real comparisons depend only on rational quadratic forms,
square roots, exponentials, and theta masses, which are efficiently computable as above.
No rational Euclidean factorization need exist: $\Gamma=[2]$ has no
$G\in\Q^{1\times1}$ satisfying $G^TG=\Gamma$.
\end{proof}

\begin{lemma}
\label{lem:oned-exact}
The one-dimensional sampling step used in the BLPRS construction has an exact Las Vegas
implementation from fair random bits for the parameters arising from rational Gram data.  Its expected bit
cost is polynomial in the input length.  Its branches requiring adaptive precision satisfy the
conditional tail bound in \Cref{lem:truncate-sampling} after every transcript prefix that has not
exceeded the representation cutoff, and only a constant number of lazy uniforms are stored
simultaneously.
\end{lemma}
\begin{proof}
Write $n$ for the ambient rank and $P$ for the bit length of the rational Gram data.
Karney represents a continuous uniform variable by a finite binary prefix with an unread uniform
tail.  In Algorithms E and N, after the discrete envelope and sign have been chosen, the
exponential and normal outputs have the form $k+U$, up to a fixed sign, where $U$ is this uniform
fractional tail~\citep[Algorithms~E and N]{Karney16}.
Each one-sided restriction of the continuous normal used by BLPRS is sampled by rejecting
Karney normals outside the required half-line.  In the BLPRS one-dimensional calls, the normalized
width is at least $\sqrt{\ln(2n+4)/\pi}$, and the truncation point is either
$c\in[0,1)$ or $c-1\in(-1,0]$~\citep[Section~5.1]{BLPRS13}.  The retained half-line therefore
has probability bounded below by an absolute constant, so this rejection has constant expected
cost.

Once a rejection decision has been certified from finite prefixes, the unread fractional tail of
an accepted value remains uniform on its current dyadic interval.  If its output is $a+bU$ and the
current cell for $U$ has length $h$, the conditional density is $(|b|h)^{-1}$, while $\ell$
additional bits give an output interval of width $|b|h2^{-\ell}$.  Their product is exactly
$2^{-\ell}$, independently of how many bits were read before the current prefix.

The comparison routines use balanced refinement: a deterministic threshold is certified to the
current interval width before another tail bit is read, and two random operands are refined so that
their current output intervals differ by at most a factor $2^{\poly(n,P)}$.  Before the
representation cutoff, all scales, reciprocal scales, and derivatives used to propagate an
interval have polynomial bit length, so this balance costs only a polynomial number of refinement
steps.  For two operands, condition on the one with the smaller current output interval and use the
independent tail with the larger interval as $X$.  Consequently the density of $X$ times the sum of
the two intervals after $\ell$ further rounds is at most $2^{\poly(n,P)-\ell}$.  The same estimate
holds against a fixed certified threshold.  Thus \Cref{lem:lazy-comparison} gives a polynomial
$q_0$ and the conditional $6\cdot2^{-q}$ bound required by \Cref{lem:truncate-sampling}.

The acceptance probabilities and theta masses in the BLPRS rejection sampler are efficiently
computable from the rational Gram data by certified intervals.  For the one-dimensional acceptance
ratio, the adjacent proposal point $y$ satisfies
$a=\pi(y^2-x^2)/r^2\ge0$.  Draw an independent exact exponential $E$ and accept exactly
when $E\ge a$; this has probability $e^{-a}$.  The exponential tail used in this comparison is
independent of the proposed point.  Outside the standard Gaussian and exponential tail event, all
continuous values have polynomial-bit magnitude; on this cutoff range, the expression $a$ and its
derivatives have size at most $2^{\poly(n,P)}$.  Balanced refinement keeps the certified interval
for $a$ within a $2^{\poly(n,P)}$ factor of the current interval for $E$, so the same relative
small-ball estimate applies.  Exceeding the representation cutoff is accounted for separately by
$\delta_{\mathrm{repr}}$ in \Cref{lem:truncate-sampling}.

Replacing each remaining real comparison by the preceding lazy comparison therefore preserves its
exact discrete distribution.
The inner iteration succeeds with probability at least $1/2$, and the outer rejection step succeeds
with probability at least $e^{-2}$~\citep[Lemma~2.3 and Section~5]{BLPRS13}.  Hence the expected
number of iterations, and thus the expected bit cost, is polynomial.  Each iteration discards its
lazy variables before the next one begins.
\end{proof}

\subsection{Implementing the ADRS sampler with fair random bits}

The remaining random choices in ADRS are integer/rational category counts, pairings, permutations,
and Bernoulli or Poisson variables.  Uniform integers, permutations, and the primitive vector $z$
are generated by ordinary rejection from fixed-length random strings.  Bernoulli decisions use
\Cref{lem:lazy-bernoulli}.

For a Poisson variable of rational mean $\lambda>0$, put
$b=\lceil8\lambda\rceil+2\kappa$ and $k_0=\lfloor\lambda\rfloor$.  Starting from $w_{k_0}=1$,
compute the unnormalized weights on $\{0,\ldots,b\}$ by
\[
 w_{k-1}=w_k\frac{k}{\lambda}\quad(k\le k_0),
 \qquad
 w_{k+1}=w_k\frac{\lambda}{k+1}\quad(k\ge k_0).
\]
Every multiplier is at most one.  Fixed-point interval arithmetic with
$Q=4\kappa+4\lceil\log_2(b+2)\rceil+20$ guard bits therefore produces dyadic weights whose normalized
distribution differs by $\exp(-\Omega(\kappa))$ from the Poisson distribution conditioned on
$X\le b$.  Indeed, the recurrence accumulates at most $b2^{-Q}$ error in any weight and hence at most
$b^2 2^{-Q}$ in total; the exact mode weight is one, so normalization does not amplify this bound.
An exact draw from the dyadic weights uses rejection from a uniform integer.  This bounds the error
from representing the truncated weights.  It remains to bound the probability that the Poisson
variable lies beyond the truncation point.  The standard Chernoff bound
$\Pr[X>b]\le(e\lambda/b)^b$ is at most $e^{-2\kappa}$ by the choice of $b$.  Thus the draw has
statistical error $\exp(-\Omega(\kappa))$ and costs
$O((\lambda+\kappa)\poly(\log(2+\lambda),\kappa))$ bit operations.  This mode-centered computation
does not form $e^{-\lambda}$ or subtract nearly equal cumulative probabilities.  The sum of the
Poisson means in an ADRS call is already charged to the number of list elements processed in its
running-time analysis.

Consequently the fair-random-bit execution has the distribution in
\Cref{thm:adrs-smooth-sampler}, up to $\exp(-\Omega(\kappa))$ statistical error, and uses expected
$2^{n/2+\polylog(\kappa)+o(n)}\poly(n,\beta_{\mathrm{in}},\beta_s)$ bit operations per call.

The discrete Gaussian has unbounded support, so coefficient lengths cannot be deterministically
polynomial without a tail event.  The following standard truncation supplies the required bound.

\begin{lemma}
\label{lem:coefficient-tail}
Let $G\in\Q^{n\times n}$ be a nonsingular basis of bit length at most $P$, let
$\sigma=s^2\in\Q_{>0}$ have bit length at most $P$, let $X\sim D_{G\Z^n,s}$, and let
$K\ge n^3$.  Except with probability $\exp(-\Omega(K))$, the ambient coordinates of $X$ and its
coefficient vector $G^{-1}X\in\Z^n$ have bit length $\poly(n,P,\log K)$.
\end{lemma}
\begin{proof}
Apply \Cref{lem:dgs-tail} with $T=\sqrt{C K/n}$ for a sufficiently large absolute constant $C$.
Because $K\ge n^3$, the polynomial factor inside the tail bound is dominated by the Gaussian
exponent, and $\norm X\le s\sqrt{CK}$ except with probability $\exp(-\Omega(K))$.  Standard
determinant and minor bounds for rational matrices give
$\log(1+\norm{G^{-1}})=\poly(n,P)$.  Hence
$\norm{G^{-1}X}\le\norm{G^{-1}}s\sqrt{CK}$
has logarithm polynomial in $n$, $P$, and $\log K$ on the same event.  Rational ambient
coordinates then have polynomial bit length as well.
\end{proof}

For CVP, we represent the augmented lattice by a rational Gram matrix.  The same tail argument
controls its coefficient representation.

\begin{lemma}
\label{lem:coefficient-tail-gram}
Let $\Gamma\in\Q^{r\times r}$ be positive definite of bit length at most $P$, let
$\sigma=s^2\in\Q_{>0}$ have bit length at most $P$, and let $a\in\Z^r$ be the coefficient vector of
$X\sim D_{\Lambda,s}$ in the lattice represented by $\Gamma$.  If $K\ge r^3$, then, except with
probability $\exp(-\Omega(K))$,
\[
a^T\Gamma a\le CKs^2
\qquad\text{and}\qquad
\norm a^2\le CKs^2/\underline\lambda,
\]
where $C$ is an absolute constant and
$0<\underline\lambda\le\lambda_{\min}(\Gamma)$ is a computable rational of bit length
$\poly(r,P)$.  In particular, $a$ has bit length $\poly(r,P,\log K)$.
\end{lemma}
\begin{proof}
Apply \Cref{lem:dgs-tail} with $T=\sqrt{CK/r}$ to obtain
$a^T\Gamma a=\norm X^2\le CKs^2$ outside an event of probability $\exp(-\Omega(K))$.
Positive definiteness gives $a^T\Gamma a\ge\lambda_{\min}(\Gamma)\norm a^2$.  Clearing
denominators in $\Gamma$, and using its positive determinant together with a rational upper bound
on its largest eigenvalue, gives a rational lower bound
$\underline\lambda\le\det(\Gamma)/\lambda_{\max}(\Gamma)^{r-1}\le\lambda_{\min}(\Gamma)$
of bit length $\poly(r,P)$.
\end{proof}

The preceding lemmas implement the random choices of one ADRS call and control the representation
tails of its output.  We now collect these ingredients in the bit-complexity guarantee used by the
outer algorithms.

\begin{theorem}[Turing-machine implementation of the ADRS sampler]
\label{thm:adrs-rational}
Let a rank-$r$ lattice be given by a rational basis or a positive-definite rational Gram matrix of
bit length at most $P$, let $s^2\in\Q_{>0}$ have bit length $\beta_s$, and let
$\kappa=\Omega(r)$.  The ADRS procedure has a randomized Turing-machine implementation whose
output, including $\bot$, is within $\exp(-\Omega(\kappa))$ statistical distance of the ideal list
distribution in \Cref{def:variable-dgs}.  There are computable bounds
$T_{\mathrm{ADRS}}(r,P,\beta_s,\kappa)$ and $S_{\mathrm{ADRS}}(r,P,\beta_s,\kappa)$, both at most
\[
2^{r/2+\polylog(\kappa)+o(r)}\poly(r,P,\beta_s),
\]
such that the expected running time is at most $T_{\mathrm{ADRS}}$ and every execution uses at most
$S_{\mathrm{ADRS}}$ bits.  The precision and representation cutoffs are also computable from the input.
\end{theorem}
\begin{proof}
Set
\[
K_{\mathrm{tail}}=C_K\bigl(r^3+\kappa+\log(2+P+\beta_s)\bigr).
\]
Let $B_{\mathrm{repr}}:=p_{\mathrm{repr}}(r,P,\beta_s,\log K_{\mathrm{tail}})$ be the
computable polynomial obtained by tracing
\Cref{lem:coefficient-tail,lem:coefficient-tail-gram,lem:kernel-basis} and the normal and
exponential tails above.  Let
$Q_0=p_{\mathrm{cmp}}(r,P,\beta_s,B_{\mathrm{repr}})$ bound the precision needed by the
certified comparisons in \Cref{lem:oned-exact}.  The quantity $K_{\mathrm{tail}}$ controls a tail
probability; a stored operand may still require $\poly(P)$ bits.

The ADRS procedure is explicit \citep[proof of Theorem 5.11]{ADRS15}.  With
\[
a=\left\lceil\frac r2+\frac{Cr}{\log r}\right\rceil,
\qquad
\ell=C\left\lceil\log^4r\right\rceil,
\qquad
M=(C\kappa^4)^{\ell+1}2^{a},
\]
it builds a tower $(\Lambda_0,\ldots,\Lambda_\ell)$ of index $2^{a}$ in $\poly(r)$ operations, draws
$M$ samples on $\Lambda_0$ in $(2^{O(r/\log r)}+M)\poly(r)$ operations, performs $\ell$ combiner
passes over a list of length $M$ in $M\poly(r,\kappa,\ell)$ operations, and repeats the whole
procedure $\kappa$ times.  The abstract operation count is therefore
\[
T_{\mathrm{abs}}
=
M\poly(r,\kappa,\ell)
=
2^{r/2+\polylog(\kappa)+o(r)},
\]
since $2^{a}=2^{r/2+O(r/\log r)}=2^{r/2+o(r)}$ and, for $\kappa=\Omega(r)$,
$(C\kappa^4)^{\ell+1}=2^{O(\log^4r\cdot\log\kappa)}=2^{\polylog(\kappa)}$.

Every operation in this procedure is an exact rational basis or Gram--Schmidt manipulation, a
one-dimensional BLPRS sample, or a category count, pairing, permutation, or Bernoulli/Poisson draw
on the current list.  The preceding subsections implement each of these on random bits, and
\Cref{lem:gram-computation,lem:oned-exact} cover the Gram computations and one-dimensional
sampling steps derived from either input representation.  Each branch requiring adaptive precision uses the unread
uniform tails described in \Cref{lem:truncate-sampling}, and the procedure stores only a polynomial
number of these tails at once.  Coefficient bounds come from \Cref{lem:coefficient-tail} for a
rational basis and from \Cref{lem:coefficient-tail-gram} for Gram coordinates; \Cref{lem:kernel-basis} bounds the
superlattice bases.  ADRS build the tower by halving coordinates~\citep{ADRS15}, so each of its
$\ell+1$ bases costs only $O(\lfloor\ell a/r\rfloor)=O(\ell)$ additional bits.  Hence
every number stored by the procedure has bit length
$\poly(r,P,\beta_s,\log K_{\mathrm{tail}})=B_{\mathrm{repr}}$ outside an event of probability
$\exp(-\Omega(K_{\mathrm{tail}}))$.  Since $\log K_{\mathrm{tail}}=O(\log r+\log\kappa+\log\log(2+P+\beta_s))$,
the factor $\poly(\log K_{\mathrm{tail}})$ is absorbed by the $2^{\polylog(\kappa)}$ already present.

Implementing each operation in the bit model multiplies $T_{\mathrm{abs}}$ by a
$\poly(r,P,\beta_s,\log K_{\mathrm{tail}})$ factor, giving the computable integer upper bound
$T_{\mathrm{ADRS}}$.  Set
\[
H_{\mathrm{ADRS}}
=\left\lceil2^{K_{\mathrm{tail}}}T_{\mathrm{ADRS}}\right\rceil,
\qquad
Q=Q_0+\left\lceil\log_2(6H_{\mathrm{ADRS}})\right\rceil+K_{\mathrm{tail}}.
\]
The stored list of length $M$ uses more space than all other exact data.  Adding the
polynomial number of stored prefixes, their $Q$-bit certified intervals, and the
$O(\log H_{\mathrm{ADRS}})$-bit operation counter does not change the stated
$S_{\mathrm{ADRS}}$ bound.

The cutoffs used here are computed deterministically from the input.  The computation evaluates the
displayed formulas for
$K_{\mathrm{tail}},a,\ell,M,T_{\mathrm{abs}},H_{\mathrm{ADRS}}$, propagates numerator and
denominator lengths through the fixed rational-arithmetic and Gram--Schmidt routines, and applies
the certified comparison moduli from \Cref{lem:truncate-sampling,lem:oned-exact}.  It returns
$(B_{\mathrm{repr}},Q,T_{\mathrm{ADRS}},S_{\mathrm{ADRS}},H_{\mathrm{ADRS}})$.  All recurrences
have polynomial length and use integer arithmetic on polynomial-bit inputs.  The resulting tuple is
computed within the time needed to write its binary representation.

The definition of $Q_0$ uniformly dominates the polynomial $q_0$ in
\Cref{lem:truncate-sampling} for every admissible primitive call.  Apply that lemma with
$K=K_{\mathrm{tail}}$.  The probabilities of exceeding $H_{\mathrm{ADRS}}$ operations or
precision $Q$ together contribute at most $2^{1-K_{\mathrm{tail}}}$ statistical distance,
uniformly over all adaptive prefixes.  The probability that a coefficient or stored number exceeds
the stated size is $\exp(-\Omega(K_{\mathrm{tail}}))$.  The procedure aborts before this occurs.
Adding these events to the $\exp(-\Omega(\kappa))$ error of the ADRS procedure proves the
claimed statistical-distance and space bounds.
\end{proof}

\begin{corollary}
\label{cor:kappa-absorb}
Fix $d\ge1$ and let $\kappa=C\bigl(r^3+\log(2+P+\beta_s)\bigr)$, which satisfies the hypothesis
$\kappa=\Omega(r)$ of \Cref{thm:adrs-rational}.  Then
$2^{O((\log\kappa)^d)}=2^{o(r)}\poly(P,\beta_s)$, and consequently
\[
T_{\mathrm{ADRS}},\;S_{\mathrm{ADRS}}
\;=\;
2^{r/2+o(r)}\poly(r,P,\beta_s).
\]
\end{corollary}
\begin{proof}
Here $\log\kappa=O\bigl(\log r+\log\log(2+P+\beta_s)\bigr)$, so
$(\log\kappa)^d=O((\log r)^d)+O\bigl((\log x)^d\bigr)$ with $x:=\log(2+P+\beta_s)$.  The first term
contributes $2^{O((\log r)^d)}=2^{o(r)}$.  For the second, $(\log x)^d=o(x)$, so it contributes
$2^{o(x)}\le2^{O(x)}=\poly(P,\beta_s)$.
\end{proof}

\subsection{Complexity of the sampling algorithm}

\Cref{thm:adrs-rational} implements each ADRS call on rational inputs.  It remains to sum the costs
and statistical errors over the calls made during one iteration of \Cref{alg:one-dgs}.

\paragraph{Expected running time.}

For \Cref{thm:one-dgs}, the adaptive execution is not stopped after a prescribed number of
operations, as it is for exact SVP.  Each iteration has expected running time
$2^{n/2+o(n)}\poly(n,\beta_{\mathrm{in}},\beta_s)$ and succeeds with probability
$1-\exp(-\Omega(\kappa_0))$, where $\kappa_0$ is defined in \Cref{alg:one-dgs}.  A local
precision or bit-length cutoff makes the current ADRS call return $\bot$, which
\Cref{alg:superlattice-run} treats as returning an empty sequence.  Only an iteration in which every run fails
restarts the outer loop.  The probability mass removed by the cutoffs is $\exp(-\Omega(\kappa_0))$, so these failures
change the output distribution only negligibly.  Repeating the iteration therefore has the same expected
running time.  Thus \Cref{thm:one-dgs} gives expected, rather than worst-case, running time.

\paragraph{Exact lattice arithmetic.}

A sample from $L_z$ is stored as an integer coefficient vector relative to the rational basis
$BC_z^{-T}$.  Its ambient coordinates are rational.  Membership in $L$ is equivalent to
$B^{-1}x\in\Z^n$ and is decided exactly.  Candidate squared norms are evaluated as exact
rationals.  On the complement of the negligible tail event in \Cref{lem:coefficient-tail}, all
operands fit in $\poly(n,\beta_{\mathrm{in}},m,\beta_s)$ bits.

\paragraph{Uniformity.}

One iteration of \Cref{alg:one-dgs} makes a polynomial number of sampler calls.  Every squared scale
and every superlattice basis has polynomial bit length, even when the superlattice index is large.
With $\kappa_0$ as defined in \Cref{alg:one-dgs}, the
$2^{\polylog(\kappa_0)+o(n)}$ overhead in \Cref{thm:adrs-rational} is absorbed into
$2^{o(n)}\poly(n,\beta_{\mathrm{in}},\beta_s)$.  A union bound over the calls in one iteration leaves
total statistical and representation error $\exp(-\Omega(n^3))$.
\section{Applications}

The pointwise estimate in \Cref{cor:pointwise-list} is explicit except for $\rho_t(L)$.  The next
lemma bounds this mass in terms of $t/\lambda_1(L)$.  Balancing a candidate vector's Gaussian
weight against this mass determines the exponents for both applications.

The constant in this mass bound comes from the Kabatiansky--Levenshtein spherical-code exponent at angle
$\pi/3$.

Put $v=\sqrt3/2$ and define
\[
\ckl
:=
\frac{1+v}{2v}\log_2\frac{1+v}{2v}
-
\frac{1-v}{2v}\log_2\frac{1-v}{2v},
\qquad
\betakl:=2^{\ckl}.
\]
The value $\ckl\approx0.4014$ is the direct Levenshtein
spherical-code exponent at angle $\pi/3$~\citep{KL78}.  We use it through the lattice-point
bound of Pujol and Stehl\'e~\citep[Lemma~3]{PS09}.  Set
\[
\mu:=\frac{\betakl^2}{2e\ln2}\approx0.4629,
\qquad
\delta:=\frac\mu2,
\qquad
c_0:=\frac12+\delta\approx0.7315.
\]

We use the following effective form of ADRS Lemma~4.2, obtained by retaining the unrounded
Kabatiansky--Levenshtein exponent and the lattice-shell argument of
\citep[Lemma~3]{PS09}.

\begin{lemma}
\label{lem:adrs-mass}
There is a computable function
$h_{\mathrm{KL}}:\Z_{>0}\to\Z_{\ge0}$ with
$h_{\mathrm{KL}}(n)=o(n)$ such that, for every rank-$n$ lattice $L$ and every $s>0$,
\[
\rho_s(L)
\le
1+
\left(
\frac{\betakl^2s^2n}{2\pi e\lambda_1(L)^2}
\right)^{n/2+1}
2^{h_{\mathrm{KL}}(n)}.
\]
The function $h_{\mathrm{KL}}$ may be chosen computable.
\end{lemma}
\begin{proof}
The mass inequality is the argument of ADRS Lemma~4.2 with the unrounded spherical-code
exponent.  \Cref{app:effective-kl} gives a finite-dimensional version, defines
$h_{\mathrm{KL}}(n)$, and proves that it is computable in polynomial time.
\end{proof}

\begin{theorem}[{\citealp[Theorem~2.15]{ADRS15}}]
\label{thm:adrs-reduction}
Given a rational basis of a rank-$n$ lattice $L$, for every $r\ge2$ an
$r^{n/r}$-reduced basis can be computed in $\exp(O(r))\poly(n,\beta_{\mathrm{in}})$ time.
For $r=2$, the first basis vector has length $d$ satisfying
$\lambda_1(L)\le d\le2^{n/2}\lambda_1(L)$.
\end{theorem}

\subsection{Exact CVP}

For a fixed rational $\alpha\ge0$, let $\CVP_\alpha$ denote exact closest-vector search under the
promise $\dist(y,L)\le\alpha\lambda_1(L)$.  No uniqueness is assumed.

The centered algorithm rejects samples outside $L$.  For a target $y$, we instead embed the lattice
in one additional dimension and consider vectors whose last coordinate is fixed.  Their sampling
probability depends on $\dist(y,L)$, while the added coordinate is chosen so that the augmented
lattice remains above smoothing.

\begin{lemma}
\label{lem:shifted-list}
Fix $u>0$ and let $m=m_*(L,u)$ be the integer in \eqref{eq:mstar}.  Sample $z$ as in
\Cref{thm:random-smoothing}, put $\Lambda=L_z$, and let
$\mathcal E$ be the event $\rho_{1/u}(\Lambda^*\setminus\{0\})\le1/4$.
For $y\in\Q^n$, put $S=\sqrt2u$, $H=S/2$, and
\[
\widetilde\Lambda=\{(x-ky,kH):x\in\Lambda,\ k\in\Z\}.
\]
Then the following hold.
\begin{enumerate}[label=(\roman*),leftmargin=2.3em]
\item $\Pr_z[\mathcal E]\ge3/4$, and on $\mathcal E$ we have
$S>\sqrt2\eta_{1/2}(\widetilde\Lambda)$.
\item For every $x\in L$, writing $\Delta=\norm{x-y}$, the joint probability that $\mathcal E$
occurs and that a list of $N_{n+1}$ independent samples from $D_{\widetilde\Lambda,S}$ contains
$(x-y,H)$ is at least
\[
\frac1{C(2m-1)}
\min\left\{1,\frac{\exp(-\pi\Delta^2/(2u^2))}{\rho_u(L)}\right\},
\]
where $C=512$.  Replacing the ideal list by an output distribution within total variation distance
$\nu$ changes this probability by at most $\nu$; in particular, this applies to the ADRS
implementation in \Cref{thm:adrs-rational}.
\item The lattice $\widetilde\Lambda$ has a positive-definite rational Gram matrix whose bit length
is polynomial in the input length and $m$.
\end{enumerate}
\end{lemma}
\begin{proof}
The dual lattice is
\[
\widetilde\Lambda^*
=
\left\{\left(w,\frac{j+\ip wy}{H}\right):w\in\Lambda^*,\ j\in\Z\right\}.
\]
At dual parameter $1/u$, the second coordinate contributes
$\exp(-2\pi(j+\ip wy)^2)$.  Since a shifted one-dimensional Gaussian has no more mass than the
centered one, with $\vartheta_0:=\rho_{1/\sqrt2}(\Z)<1.004$,
\[
\rho_{1/u}(\widetilde\Lambda^*\setminus\{0\})
\le(\vartheta_0-1)+\vartheta_0/4<1/2.
\]
Thus $S>\sqrt2\eta_{1/2}(\widetilde\Lambda)$.

The same one-dimensional comparison gives
\[
\rho_S(\widetilde\Lambda)
=\sum_{k\in\Z}e^{-\pi k^2/4}\rho_S(\Lambda-ky)
\le\rho_2(\Z)\rho_S(\Lambda),
\]
where the inequality follows from \Cref{lem:shifted-mass}.
The vector $(x-y,H)$ has weight
$e^{-\pi/4}\exp(-\pi\Delta^2/S^2)$.  There are $N_{n+1}$ ideal samples, and
\Cref{thm:mass-comparison} bounds
$\rho_S(\Lambda)$ by $40(2m-1)2^{n/2}\rho_u(L)$.  The number of samples cancels the factor
$2^{n/2}$; the event in \Cref{thm:random-smoothing} has constant probability.
The standard inequality $1-(1-p)^N\ge\tfrac12\min\{1,Np\}$ proves the probability bound.
Indeed, using $\rho_2(\Z)<9/4$ and $e^{\pi/4}<9/4$, the reciprocal of the resulting
absolute constant is at most
$320\rho_2(\Z)e^{\pi/4}/(3\sqrt2)<512$.
The total-variation statement follows because the probability of any event changes by at most the
total-variation distance.

If $G$ is the rational basis of $\Lambda$, the natural basis of $\widetilde\Lambda$ has Gram entries
$G^TG$, $-G^Ty$, and $\norm y^2+H^2$, all rational because $H^2=u^2/2$.  Coefficients recover the
integer $k$ in the last coordinate and test membership in $L$ exactly.  The representation and tail bounds follow
from \Cref{lem:gram-computation,lem:coefficient-tail-gram,thm:adrs-rational}.
\end{proof}

Define
\[
\alpha_c:=\frac{\sqrt{2e}}{\betakl},
\qquad
\psi(\alpha):=
\begin{cases}
\mu\alpha^2/2,&0\le\alpha\le\alpha_c,\\[0.4ex]
1/(2\ln2)+\log_2(\alpha/\alpha_c),&\alpha\ge\alpha_c.
\end{cases}
\]

To make the repetition count computable, put
\[
R_\alpha:=
\begin{cases}
\sqrt{2\pi e}/\betakl,&0\le\alpha\le\alpha_c,\\
\sqrt\pi\,\alpha,&\alpha\ge\alpha_c,
\end{cases}
\qquad
a_{\mathrm{KL}}:=\frac{\betakl^2}{2\pi e}.
\]
Fix integers $U_\alpha>\max\{1,R_\alpha\}$ and
\[
C_\alpha\ge12+\max\{0,\log_2(4a_{\mathrm{KL}}R_\alpha^2)\}.
\]
Since $\alpha$ is fixed, certified upper approximations can be used to choose and hardwire both
integers; no exact comparison with an integer boundary is required.
The constant accounts for the absolute constant in the preceding lemma, the extra power in
\Cref{lem:adrs-mass}, and rounding to the scale grid.

For each $n$, certified interval evaluation of the fixed exponent $\psi(\alpha)$ gives an
integer $E_\alpha(n)$ satisfying
\begin{equation}
\label{eq:E-alpha}
 \psi(\alpha)n+h_{\mathrm{KL}}(n)+C_\alpha
 \le E_\alpha(n)
 \le \psi(\alpha)n+h_{\mathrm{KL}}(n)+C_\alpha+2.
\end{equation}
It suffices to enclose $\psi(\alpha)n$ in an interval of width less than one and round the upper
endpoint upward.  This takes polynomially many bit operations and does not require deciding whether
the real number on the left of \eqref{eq:E-alpha} is an integer.

Because $\lambda_1(L)$ is unknown, the estimate in \Cref{thm:adrs-reduction} is used to form a
polynomial-size grid containing a scale within a factor $1+n^{-2}$ of the optimizer.

\begin{algorithm}[H]
\caption{Exact CVP}
\label{alg:cvp}
\begin{algorithmic}[1]
\Require Nonsingular rational basis $B\in\Q^{n\times n}$ of $L$, target $y\in\Q^n$,
and fixed rational $\alpha\ge0$
\State $\beta_{\mathrm{in}}\gets\operatorname{bits}(B)$ and
$\beta_y\gets\operatorname{bits}(y)$
\State use \Cref{thm:adrs-reduction} to obtain its first vector $b_1$ and set $d^2\gets\norm{b_1}^2\in\Q$
\State $J_{\mathrm{sc}}\gets\min\{J\ge1:(1+n^{-2})^J\ge2^n(n+1)U_\alpha^2\}$
\State $u_j^2\gets(n+1)^2U_\alpha^2d^2(1+n^{-2})^{-2j}$ for $0\le j\le J_{\mathrm{sc}}$
\State compute $E_\alpha(n)$ as in \eqref{eq:E-alpha}
\State $\mathsf{best}\gets0$ and
$\kappa_0\gets\left\lceil C_\kappa(n^3+\log(2+\beta_{\mathrm{in}}+\beta_y))\right\rceil$
\For{$j=0,\ldots,J_{\mathrm{sc}}$}
  \State $J_j\gets J_{\mathrm{mod}}(B,u_j^2)$
  \State $A_j\gets(2J_j+1)2^{E_\alpha(n)+2}$
  \For{$m=1,\ldots,J_j$}
    \For{$r=1,\ldots,A_j$}
      \State sample $z$ uniformly from the primitive vectors in $(\Z/2^m\Z)^n$
      and construct $C_z$ as in \Cref{lem:kernel-basis}; set $G\gets BC_z^{-T}$ and $L_z\gets G\Z^n$
      \State construct the Gram matrix of
      $\widetilde\Lambda=\{(x-ky,ku_j/\sqrt2):x\in L_z,\ k\in\Z\}$
      \State $W\gets$ the rank-$(n+1)$ ADRS output with squared parameter $2u_j^2$ and statistical parameter $\kappa_0$
      \If{$W\ne\bot$}
        \For{each coefficient vector $(a,k)\in W$ with $k=1$}
          \State $x\gets Ga$; if $x\in L$ and $\norm{x-y}^2<\norm{\mathsf{best}-y}^2$, set $\mathsf{best}\gets x$
        \EndFor
      \EndIf
    \EndFor
  \EndFor
\EndFor
\State \Return $\mathsf{best}$
\end{algorithmic}
\end{algorithm}

The displayed procedure describes the untruncated execution.  The randomized-Turing-machine
implementation analyzed below stops it at the explicit global operation and representation bounds given
in the proof and returns failure if either bound is exceeded.

\begin{theorem}[Exact CVP]
\label{thm:cvp-alpha}
For every fixed rational $\alpha\ge0$, there is a randomized Turing-machine algorithm for
$\CVP_\alpha$ on a rank-$n$ lattice specified by a nonsingular rational basis and a rational target
$y\in\Q^n$.  It succeeds with probability at least $2/3$.  On every execution, its running time is
\[
2^{(1/2+\psi(\alpha))n+o(n)}\poly(n,\beta_{\mathrm{in}},\beta_y)
\]
and its space usage is
$2^{n/2+o(n)}\poly(n,\beta_{\mathrm{in}},\beta_y)$.
In particular, the running time is strictly below $2^n$ for
$\alpha<1/\sqrt\mu$, where $1/\sqrt\mu\approx1.4697$.
\end{theorem}
\begin{proof}
Let $x$ be a closest lattice vector and write $u=R\lambda_1(L)/\sqrt n$.  Combining
\Cref{lem:shifted-list} with \Cref{lem:adrs-mass} shows that the reciprocal of the probability of
outputting $(x-y,H)$ has exponential rate at most
\begin{equation}
\label{eq:cvp-rate}
\frac{\pi\alpha^2}{2R^2\ln2}
+\frac12\max\left\{0,\log_2\!\left(\frac{\betakl^2R^2}{2\pi e}\right)\right\}.
\end{equation}
For $\alpha\le\alpha_c$, this expression is minimized at
$R=\sqrt{2\pi e}/\betakl$; for $\alpha\ge\alpha_c$, it is minimized at
$R=\sqrt\pi\alpha$.  The two values are precisely the two branches of $\psi(\alpha)$.

The algorithm does not know $\lambda_1(L)$.  Use \Cref{thm:adrs-reduction} to obtain its first
vector $b_1$, and write $d=\norm{b_1}$ in the analysis.  The algorithm stores only the rational
value $d^2$.  Try the squared scales
\[
u_j^2=(n+1)^2U_\alpha^2d^2(1+n^{-2})^{-2j}
\]
until $(1+n^{-2})^j\ge2^n(n+1)U_\alpha^2$.  The first scale is above the optimizer and the
last is below it.  This is a polynomial-size grid, and one scale is within a
factor $1+n^{-2}$ of the minimizing scale above.  At each scale use the modulus range
$1\le m\le J_{\mathrm{mod}}(B,u_j^2)$ from \eqref{eq:explicit-J}.  For every pair $(j,m)$, sample a
fresh primitive vector $z$, form the Gram matrix in \Cref{lem:shifted-list}, run the ADRS sampler at
$\sqrt2u_j$, and consider vectors with last-coordinate coefficient $k=1$ whose associated point lies in $L$.  Exact rational
comparison of $\norm{x-y}^2$ keeps the best candidate.

For the selected grid point, $R_\alpha\le R\le(1+n^{-2})R_\alpha$.  Moving from
$R_\alpha$ to $R$ increases $n$ times the expression in \eqref{eq:cvp-rate} by at most
$n\log_2(1+n^{-2})\le1$.  The extra power $n/2+1$ in \Cref{lem:adrs-mass} contributes at
most $\max\{0,\log_2(4a_{\mathrm{KL}}R_\alpha^2)\}$ beyond that rate.  Together with the
factor $512$ in \Cref{lem:shifted-list} and $1+x\le2\max\{1,x\}$, the definition of
$C_\alpha$ therefore gives
\[
p_x
\ge
\frac{2^{-\psi(\alpha)n-h_{\mathrm{KL}}(n)-C_\alpha}}{2J_{\mathrm{mod}}(B,u_j^2)+1}.
\]
The value $A_j$ in \Cref{alg:cvp} therefore makes the probability of missing a fixed closest
vector at most $e^{-4}$.  If it is seen, exact distance comparison returns a closest vector even
when there are several.

There are $2^{\psi(\alpha)n+h_{\mathrm{KL}}(n)+o(n)}\poly(n,\beta_{\mathrm{in}},\beta_y)
=2^{\psi(\alpha)n+o(n)}\poly(n,\beta_{\mathrm{in}},\beta_y)$ calls.  To specify the global cutoff, put
\[
 K_{\mathrm{calls}}:=\sum_{j=0}^{J_{\mathrm{sc}}}J_jA_j.
\]
Let $P_{\max}$ bound the bit lengths of every superlattice basis and augmented Gram matrix in these
calls, and let $\beta_{\sigma,\max}$ bound the bit lengths of their squared parameters.  The basis
construction above and \Cref{lem:kernel-basis} give computable bounds polynomial in
$n,\beta_{\mathrm{in}},\beta_y$.  Let
$T_{\max}:=T_{\mathrm{ADRS}}(n+1,P_{\max},\beta_{\sigma,\max},\kappa_0)$ and
$S_{\max}:=S_{\mathrm{ADRS}}(n+1,P_{\max},\beta_{\sigma,\max},\kappa_0)$ denote the
expected-time bound and the space bound on every execution from \Cref{thm:adrs-rational}.  If
$W_{\mathrm{other}}$ is a computable upper bound on the expected cost of sampling the primitive
vectors, constructing the bases, and executing the remaining finite loops, define
\[
 W_{\max}:=K_{\mathrm{calls}}T_{\max}+W_{\mathrm{other}}.
\]
These quantities have the time and space bounds in the theorem.  Stop the complete procedure after
$100W_{\max}$ bit operations.  Markov's inequality bounds the probability of reaching this time limit
by $1/100$.  The precision and representation cutoffs belong to the individual ADRS calls; a violation
returns $\bot$, and their total contribution is $\exp(-\Omega(n^3))$ by
\Cref{lem:adaptive-tv,thm:adrs-rational}.  Calls are sequential, so the space used is
$S_{\max}$ plus polynomial bookkeeping.  The stopped procedure therefore succeeds with probability
greater than $2/3$ and has the stated time and space bounds.
\end{proof}

At $\alpha=1$, the exponent is $c_0$.  In the BDD regime $0\le\alpha<1/2$, re-optimizing the proof
of ADRS Corollary~7.4~\citep{ADRS15} and repeating the ADRS call when necessary gives exponent
$\max\{1/2,2\ckl\alpha^2/(1-4\alpha^2)\}$.  It is smaller than our exponent for
$\alpha<0.42716$, while \Cref{thm:cvp-alpha} is smaller above that point.  BDD ceases to apply at
$\alpha=1/2$, whereas \Cref{thm:cvp-alpha} continues without uniqueness until its exponent reaches
$1$ at $\alpha=1/\sqrt\mu$.  For larger $\alpha$, the unrestricted $2^{n+o(n)}$ algorithm of
Aggarwal, Dadush, and Stephens-Davidowitz is faster~\citep{ADS15}.

\begin{corollary}[$\gamma$-uSVP]
For every fixed rational $\gamma>1$, exact SVP on rank-$n$ lattices given by rational bases and satisfying
$\lambda_2(L)>\gamma\lambda_1(L)$ is solvable with probability at least $2/3$ in
time
$2^{(1/2+\mu/(2\gamma^2))n+o(n)}\poly(n,\beta_{\mathrm{in}})$ and space
$2^{n/2+o(n)}\poly(n,\beta_{\mathrm{in}})$.
\end{corollary}
\begin{proof}
Choose a constant prime $p>\gamma+1$ and, for each basis coordinate $i$, let
$L_{i,p}=\{Bz:z_i\equiv0\pmod p\}$.  If $v=Ba$ is shortest, then $a$ is primitive, so for some
$i$ we have $a_i\not\equiv0\pmod p$.  With $r=a_i\bmod p$ and $y=rb_i$, the vector $v$ lies in
$y-L_{i,p}$.  A nonzero vector of $L_{i,p}$ collinear with $v$ has length at least
$p\lambda_1(L)$, while a noncollinear vector has length at least $\lambda_2(L)$.  Thus
$\lambda_1(L_{i,p})>\gamma\lambda_1(L)$.  In the coset $y-L_{i,p}$, the next collinear vector
after $v$ has length at least $(p-1)\lambda_1(L)$, and every noncollinear vector has length at
least $\lambda_2(L)$.  Hence $v$ is the unique shortest vector in the coset $y-L_{i,p}$ and
$\dist(y,L_{i,p})<\lambda_1(L_{i,p})/\gamma$.  Equivalently,
$w:=y-v$ is the unique closest point in $L_{i,p}$ to the target $y$.  Apply
\Cref{thm:cvp-alpha} with $\alpha=1/\gamma$ to the polynomially many pairs $(i,r)$; when a call
returns $w\in L_{i,p}$, keep $y-w\in L$ as a candidate and return the shortest nonzero candidate.
\end{proof}

\subsection{An exact-SVP consequence}
At the largest scale for which $\rho_t(L)$ is subexponential, a fixed shortest vector has
Gaussian weight $2^{-\delta n}$.  Repeating a $2^{n/2+o(n)}$ above-smoothing call
$2^{\delta n+o(n)}$ times therefore gives the exponent $c_0=1/2+\delta$.  The same decomposition
also shows how the bound changes for a different above-smoothing sampler.

Let $t_0=\sqrt{2\pi e}\lambda_1(L)/(\betakl\sqrt n)$.
At this scale, \Cref{lem:adrs-mass} gives
\begin{equation}
\label{eq:rho-subexp}
\rho_{t_0}(L)\le1+2^{h_{\mathrm{KL}}(n)}\le2^{h_{\mathrm{KL}}(n)+1}.
\end{equation}
Moreover,
\begin{equation}
\exp\!\left(-\frac{\pi\lambda_1(L)^2}{2t_0^2}\right)
=2^{-\mu n/2}=2^{-\delta n}.
\end{equation}
If $t\in[t_0,(1+n^{-2})t_0]$, then \Cref{lem:gaussian-scaling} changes
\eqref{eq:rho-subexp} by only a constant factor.  Hence, for a fixed shortest vector $v$,
\Cref{cor:pointwise-list} gives
\begin{equation}
\label{eq:svp-run}
\Pr[v\text{ is output in a run with random $z$ and $m=m_*$}]
\ge
\frac{2^{-\delta n-h_{\mathrm{KL}}(n)-12}}{2m_*-1}.
\end{equation}

Certified interval evaluation similarly gives, in polynomial time, an integer
$E_{\mathrm{SVP}}(n)$ satisfying
\begin{equation}
 \delta n+h_{\mathrm{KL}}(n)
 \le E_{\mathrm{SVP}}(n)
 \le \delta n+h_{\mathrm{KL}}(n)+2.
\end{equation}

The scale $t_0$ is unknown.  Apply \Cref{thm:adrs-reduction} with $r=2$ and let $d$ be the length
of its first vector.  Define
$J_{\mathrm{sc}}:=\min\{J\ge1:(1+n^{-2})^J\ge2^n(n+1)\}=O(n^3)$ and
$\tau_j=(n+1)^2d^2(1+n^{-2})^{-2j}$ for $0\le j\le J_{\mathrm{sc}}$.
All comparisons defining this grid are exact rational comparisons.  The first scale
$(n+1)d$ is larger than $t_0$, whereas the last is at most
$d/2^n\le2^{-n/2}\lambda_1(L)<t_0$.  Consecutive scales differ by $1+n^{-2}$, so some
$t_j=\sqrt{\tau_j}$ lies in $[t_0,(1+n^{-2})t_0]$.

\begin{algorithm}[H]
\caption{Exact SVP via above-smoothing DGS}
\label{alg:svp}
\begin{algorithmic}[1]
\Require Nonsingular rational basis $B\in\Q^{n\times n}$ of $L$
\State $\beta_{\mathrm{in}}\gets\operatorname{bits}(B)$
\If{$n=1$}
  \State \Return the nonzero column of $B$
\EndIf
\State use \Cref{thm:adrs-reduction} to obtain its first vector $b_1$ and set $d^2\gets\norm{b_1}^2\in\Q$
\State $J_{\mathrm{sc}}\gets\min\{J\ge1:(1+n^{-2})^J\ge2^n(n+1)\}$
\State $\tau_j\gets(n+1)^2d^2(1+n^{-2})^{-2j}$ for $0\le j\le J_{\mathrm{sc}}$
\State compute $E_{\mathrm{SVP}}(n)$ as above
\State $\mathsf{best}\gets$ any nonzero basis vector of $L$
\State $\kappa_0\gets\left\lceil C_\kappa(n^3+\log(2+\beta_{\mathrm{in}}))\right\rceil$
\For{$j=0,\ldots,J_{\mathrm{sc}}$}
  \State compute $J_j:=J_{\mathrm{mod}}(B,\tau_j)$ from \eqref{eq:explicit-J}
  \State $A_j\gets(2J_j+1)2^{E_{\mathrm{SVP}}(n)+14}$
  \For{$m=1,\ldots,J_j$}
    \For{$r=1,\ldots,A_j$}
      \State $W\gets$ the sequence returned by \Cref{alg:superlattice-run} on $(B,\tau_j,m,\kappa_0)$
      \For{each nonzero vector $x\in W$}
        \State if $\norm{x}^2<\norm{\mathsf{best}}^2$, set $\mathsf{best}\gets x$
      \EndFor
    \EndFor
  \EndFor
\EndFor
\State \Return $\mathsf{best}$
\end{algorithmic}
\end{algorithm}

As in the main sampling algorithm, every superlattice is represented by a rational basis
and no quotient is enumerated.  The displayed procedure is untruncated; the implementation stops it at
the explicit global bounds below and returns failure if a bound is exceeded.

\begin{theorem}[Exact SVP]
\label{thm:exact-svp}
There is a randomized Turing-machine algorithm that, on input a rank-$n$ lattice
$L\subseteq\R^n$ specified by a nonsingular rational basis, returns a shortest nonzero vector with
probability at least $2/3$.  On every execution, it runs in time
\[
2^{c_0n+o(n)}\poly(n,\beta_{\mathrm{in}})
\]
and uses space $2^{n/2+o(n)}\poly(n,\beta_{\mathrm{in}})$.
\end{theorem}
\begin{proof}
Consider the grid point in the interval above and its distinguished modulus $m_*$.  By
\eqref{eq:svp-run} and the definition of $A_j$, the expected number of occurrences of a fixed
shortest vector is at least $4$.  Hence the probability that all $A_j$ runs miss that vector is at
most $e^{-4}$.  Every output vector lies in $L$, and exact squared-norm comparisons keep a shortest
one once it appears.

The number of runs is
\[
2^{\delta n+h_{\mathrm{KL}}(n)+o(n)}\poly(n,\beta_{\mathrm{in}})
=2^{\delta n+o(n)}\poly(n,\beta_{\mathrm{in}}).
\]
Each call costs
$2^{n/2+o(n)}\poly(n,\beta_{\mathrm{in}})$ bit operations and uses
$2^{n/2+o(n)}\poly(n,\beta_{\mathrm{in}})$ space by
\Cref{thm:adrs-rational,cor:kappa-absorb}.  Calls are sequential, and their output is discarded
before the next call, so
the exponents add in time but not in space.  The statistical-error composition and the conversion
to a fixed global time bound are given below.

It remains to justify the two points suppressed in the preceding running-time calculation: the
sampling errors must remain negligible over exponentially many calls, and the expected cost of an
ADRS call must be controlled by a fixed cutoff.  The calls made by \Cref{alg:svp} are adaptive
only through their generated lattices and random outputs, so \Cref{lem:adaptive-tv} applies.

For \Cref{alg:svp}, let
$K_{\mathrm{calls}}:=\sum_{j=0}^{J_{\mathrm{sc}}}J_jA_j
=2^{\delta n+h_{\mathrm{KL}}(n)+o(n)}\poly(n,\beta_{\mathrm{in}})$.
Let $P_{\max}$ be the computable uniform bound from \Cref{lem:kernel-basis} on the bit length of
all superlattice bases occurring in these calls, and let $\beta_{\tau,\max}$ bound the bit lengths
of their squared scales.  Both are polynomial in $n$ and $\beta_{\mathrm{in}}$.  Consequently the
chosen
$\kappa_0=C_\kappa(n^3+\log(2+\beta_{\mathrm{in}}))$ also dominates
$n^3+\log(2+P_{\max}+\beta_{\tau,\max})$ after increasing $C_\kappa$.
The sum of all sampling, tail, and finite-precision errors is therefore
$\exp(-\Omega(n^3))$ by \Cref{lem:adaptive-tv,thm:adrs-rational}.  Replacing the ideal calls by
their implementations on a randomized Turing machine changes the probability of seeing a shortest vector by at most
this amount.

Write $T_{\mathrm{ADRS}}(n,P_{\max},\beta_{\tau,\max},\kappa_0)$ for an integer upper bound on
the concrete computable expectation supplied by \Cref{thm:adrs-rational}.  Let
$W_{\mathrm{other}}$ be an integer upper bound on the expected cost of sampling primitive vectors,
constructing the bases, scanning each output list, and performing the exact norm comparisons in
\Cref{alg:svp}.  Define
$W_{\max}:=K_{\mathrm{calls}}T_{\mathrm{ADRS}}(n,P_{\max},\beta_{\tau,\max},\kappa_0)+W_{\mathrm{other}}$.
The preceding bounds give
$W_{\max}\le2^{c_0n+o(n)}\poly(n,\beta_{\mathrm{in}})$.  Stop the complete algorithm after
$100W_{\max}$ bit operations and return failure.  Markov's inequality shows that this time limit is
reached with probability at most $1/100$.  The per-call precision and representation cutoffs already
contribute only the preceding negligible implementation error.  Together with the $e^{-4}$ failure
bound for the ideal procedure, this leaves success probability greater than $2/3$.  This proves the
time and success claims in \Cref{thm:exact-svp}; the space bound in
\Cref{thm:adrs-rational} holds on every execution.
\end{proof}

\subsection{Exact SVP from above-smoothing DGS}

The preceding exponent consists of the cost of one above-smoothing call and the number of calls
needed to output a shortest vector.  Separating these two terms gives the exact-SVP bound for a
sampler with a different number of outputs, running time, or space.

\begin{theorem}[Exact SVP from faster sampling above smoothing]
\label{thm:svp-above-smoothing}
Fix $\xi,\omega,\zeta\ge0$ with $\omega\ge\xi$.  Suppose a sampler satisfying
\Cref{def:variable-dgs} on every rank-$n$ input lattice $\Lambda$ has threshold
$\sqrt2\eta_{1/2}(\Lambda)$.  Whenever the input parameter exceeds this threshold, it produces
$N=N(n)$ samples, where $N(n)$ is computable in $2^{o(n)}$ bit operations and
$\log_2N(n)=\xi n+o(n)$, and has statistical error
$\exp(-\Omega(\kappa))$.  Suppose further that one call has computable expected bit cost
$2^{\omega n+\polylog(\kappa)+o(n)}\poly(n,\beta_{\mathrm{in}})$, including outputting and
processing the samples, and uses at most
$2^{\zeta n+o(n)}\poly(n,\beta_{\mathrm{in}})$ bits on every execution.  The sampler also supplies,
in $2^{o(n)}$ bit operations, a numerical upper bound on its expected cost and explicit precision,
representation, and space cutoffs for one call.  Then a randomized Turing-machine algorithm solves
exact SVP with probability at least $2/3$.  On every execution, its running time is
\[
2^{c(\xi,\omega)n+o(n)}\poly(n,\beta_{\mathrm{in}}),
\qquad
c(\xi,\omega)
=
\omega+\max\left\{0,\frac12+\delta-\xi\right\},
\]
and its space usage is $2^{\zeta n+o(n)}\poly(n,\beta_{\mathrm{in}})$.
\end{theorem}
\begin{proof}
For a fixed shortest vector $v$, the proof of \Cref{cor:pointwise-list} with $N$ samples
gives, at the good scale and modulus,
\[
\Pr[v\text{ occurs}]
\ge
\frac1{160(2m_*-1)}
\min\left\{1,
N2^{-n/2}\frac{e^{-\pi\lambda_1(L)^2/(2t^2)}}{\rho_t(L)}
\right\}.
\]
Write $t=R\lambda_1(L)/\sqrt n$.  Substitution of \Cref{lem:adrs-mass} shows that the two
$n$-dependent terms in the reciprocal success probability contribute
\[
\frac{\pi}{2R^2\ln2}
+
\frac12\left[\log_2\!\left(\frac{\betakl^2R^2}{2\pi e}\right)\right]_+.
\]
The minimum over $R>0$ is $\delta$, attained at
$R=\sqrt{2\pi e}/\betakl$.  The factor $N2^{-n/2}$ contributes $\xi-1/2$.
More explicitly, at each scale and modulus it suffices to make
\[
2^{12}(2J_{\mathrm{mod}}(B,t^2)+1)
\max\left\{1,\left\lceil\frac{2^{\lceil n/2\rceil+E_{\mathrm{SVP}}(n)}}{N(n)}\right\rceil\right\}
\]
independent runs.  This repetition count is computable in $2^{o(n)}$ bit operations,
and its exponential rate is $\max\{0,1/2+\delta-\xi\}$.
Multiplying by the cost per call gives the time exponent.  Sequential execution gives the stated
space, and the error-composition and stopping argument above applies unchanged.
\end{proof}

For the ADRS sampler, $\xi=\omega=\zeta=1/2$, and
\Cref{thm:svp-above-smoothing} gives $c(1/2,1/2)=c_0$.  More generally, substituting a different number of
samples or a different cost per call into the displayed formula gives the corresponding exact-SVP
exponent; the random-superlattice and Gaussian-mass bounds are unchanged.

\section{Conclusion}

\Cref{thm:one-dgs} produces one centered discrete Gaussian sample at every positive rational squared
parameter in expected $2^{n/2+o(n)}$ time and $2^{n/2+o(n)}$ space on every execution, up to polynomial
factors in the input length.  It uses random superlattices that are smooth at the requested scale
with constant probability and scans one ADRS output list for a point of the original lattice.  The
Gaussian-mass comparison gives inverse-polynomial success probability, while
\Cref{thm:mass-floor} shows that the $2^{n/2}$ factor is unavoidable in comparisons against
$\rho_{s/\sqrt2}(L)$.

A one-dimensional augmentation gives exact CVP in sub-$2^n$ time for
every fixed rational $\alpha<1/\sqrt\mu\approx1.4697$, without a uniqueness assumption.  The same
pointwise estimate gives the secondary exact-SVP and $\gamma$-uSVP consequences.

Two limitations remain.  The algorithm stores the $2^{n/2+o(n)}$ samples returned by one ADRS call.
It is also centered: for a shift, the success probability is governed by the ratio between shifted
and centered Gaussian masses, which can be arbitrarily small.  The distance guarantee controls this
ratio for the CVP application, but not for arbitrary-parameter shifted DGS.

\section*{Acknowledgments}
We are grateful to Minki Hhan for helpful discussions.

\paragraph{AI use disclosure.}
OpenAI ChatGPT and Codex were used to help develop and rigorously check proofs and to assist with
writing the manuscript.  The author independently verified all claims and proofs and takes full
responsibility for the correctness, originality, and integrity of the manuscript.
{\small
\bibliographystyle{alphaurl}
\bibliography{references}
}

\appendix
\section{An effective Gaussian-mass bound}
\label{app:effective-kl}

This appendix makes the $o(n)$ term in \Cref{lem:adrs-mass} uniform and computable.  We first
apply a finite-dimensional spherical-code bound and only then sum over radii.  The multiplicative
coronas in the asymptotic proof of \citep[Lemma~3]{PS09} suffice when the outer radius is fixed;
absolute-width shells give the uniform estimate needed here.

Write $M(d,\theta)$ for the largest size of a spherical code in $S^{d-1}$ with minimum angle
$\theta$.  For $a,b>-1$, let $t_{1,k}^{a,b}$ be the largest zero of the degree-$k$ Jacobi
polynomial.  With $\alpha=(d-3)/2$, the finite Levenshtein bound is
\[
M_{\mathrm{Lev}}(d,\theta)=
\begin{cases}
2\binom{k+d-1}{d-1},
 &t_{1,k}^{\alpha+1,\alpha}<\cos\theta\le t_{1,k}^{\alpha+1,\alpha+1},\\[1ex]
\binom{k+d-1}{d-1}+\binom{k+d-2}{d-1},
 &t_{1,k-1}^{\alpha+1,\alpha+1}<\cos\theta\le t_{1,k}^{\alpha+1,\alpha}.
\end{cases}
\]
The relevant line and the integer $k$ are uniquely determined.  Levenshtein's linear-programming
bound gives $M(d,\theta)\le M_{\mathrm{Lev}}(d,\theta)$; see
\citep[Equations~(6)--(7)]{SZ24}.  Put
\[
 S_d:=M_{\mathrm{Lev}}(d,\pi/3).
\]
For $v=\sqrt3/2$ and $c=\ckl$, the Kabatiansky--Levenshtein analysis gives
\begin{equation}
\label{eq:lev-asymptotic}
 \log_2 S_d=cd+o(d).
\end{equation}

We first dispose of rank one.  After scaling, write $L=\Z$ and $x=s/\lambda_1(L)$.  If
$0<x\le1$, then $k^2\ge k$ and the elementary inequality
$e^{-\pi/x^2}\le e^{-\pi}x^3$ give
\[
 \rho_x(\Z)-1
 \le \frac{2e^{-\pi/x^2}}{1-e^{-\pi/x^2}}
 \le x^3.
\]
If $x\ge1$, Poisson summation and the preceding estimate at $1/x$ give
$\rho_x(\Z)=x\rho_{1/x}(\Z)\le2x\le1+2x^3$.
Since $\betakl>1$ and
$2^8/(2\pi e)^{3/2}>2$, the inequality in \Cref{lem:adrs-mass} holds in rank one after setting
$h_{\mathrm{KL}}(1)=8$.  Henceforth assume $n\ge2$.

We turn this spherical bound into a lattice-point bound that holds at every radius.  Normalize
$\lambda_1(L)=1$ and let
$N_L(R)=|\{x\in L\setminus\{0\}:\norm{x}\le R\}|$.  For $R\ge1$, cover these vectors by the shells
\[
 E_j=\{x\in L:j/n\le\norm{x}<(j+1)/n\},
 \qquad j=n,\ldots,\lfloor nR\rfloor.
\]
If $x,y\in E_j$ are distinct, $a=\norm{x}$, $b=\norm{y}$, and $\phi$ is the angle between them, then
\[
 2ab(1-\cos\phi)=\norm{x-y}^2-(a-b)^2\ge1-n^{-2}.
\]
Thus the normalized vectors in $E_j$ form a spherical code whose chordal distance is at least
$q_j:=\sqrt{1-n^{-2}}/r_j$, where $r_j=(j+1)/n$.  Let $\mathrm{cap}_n(q)$ denote the normalized
area of a cap in $S^{n-1}$ with radial radius $q$.  Sidelnikov's cap reduction
\citep[Equation~(5)]{SZ24}, applied with target angle $\pi/3$, gives
\[
 |E_j|\le \frac{S_{n+1}}{\mathrm{cap}_n(q_j)}.
\]
Orthogonal projection of the cap onto $B_2^{n-1}(q)$ gives
\[
 \mathrm{cap}_n(q)\ge\frac{V_{n-1}q^{n-1}}{nV_n},
\]
where $V_d$ is the volume of the Euclidean unit ball in dimension $d$.  There are at most $2nR$
shells, $r_j\le(1+1/n)R$, $(1+1/n)^{n-1}<3$,
$(1-n^{-2})^{-(n-1)/2}<2$, and $V_n\le2V_{n-1}$.  Hence
\begin{equation}
\label{eq:effective-lattice-count}
 N_L(R)\le K_nR^n,
 \qquad K_n:=32n^2S_{n+1}.
\end{equation}

We next sum the Gaussian weights.  Put $m=n/2$.  Stieltjes integration and
\eqref{eq:effective-lattice-count} give, for every $s>0$,
\begin{equation}
\label{eq:stieltjes-mass}
 \rho_s(L)-1
 \le K_n\Gamma(m+1)\left(\frac{s^2}{\pi}\right)^m.
\end{equation}
Robbins' form of Stirling's bound implies
$\Gamma(m+1)\le2\sqrt n\,(m/e)^m$.  Let $s_0^2=2\pi/n$ and
$z=s^2n/(2\pi e)$.  If $s\ge s_0$, then $z\ge1/e$, and
\eqref{eq:stieltjes-mass} is at most $8K_n\sqrt n\,z^{m+1}$.  If $s\le s_0$, put
$x=s^2n/(2\pi)\le1$.  For every $r\ge1$,
\[
 e^{-\pi r^2/s^2}
 \le e^{-\pi/s^2+n/2}e^{-\pi r^2/s_0^2}.
\]
Applying \eqref{eq:stieltjes-mass} at $s_0$ therefore gives
\[
 \rho_s(L)-1\le2K_n\sqrt n\,e^{-m/x}.
\]
For $n\ge2$, the function on the left below is maximized at $x=m/(m+1)$, so
\[
 \frac{e^{-m/x}}{(x/e)^{m+1}}
 \le\left(1+\frac1m\right)^{m+1}\le4.
\]
Consequently, with
$P_n:=\lceil16K_n\lceil\sqrt n\rceil\rceil$, uniformly in $s$ and $L$,
\begin{equation}
\label{eq:effective-mass-pre}
 \rho_s(L)\le1+P_n
 \left(\frac{s^2n}{2\pi e\lambda_1(L)^2}\right)^{n/2+1}.
\end{equation}

It remains to extract the exponent used in the applications.  For $n\ge2$, compute a certified rational number
$\underline c_n$ such that
\[
 c-(n+2)^{-2}\le\underline c_n\le c
\]
and define
\begin{equation}
 h_{\mathrm{KL}}(n):=
 \max\left\{0,\left\lceil\log_2P_n\right\rceil
 -\left\lfloor(n+2)\underline c_n\right\rfloor\right\}.
\end{equation}
Then $2^{(n+2)c+h_{\mathrm{KL}}(n)}\ge P_n$, so
\eqref{eq:effective-mass-pre} is the inequality in \Cref{lem:adrs-mass}.  Moreover,
\eqref{eq:lev-asymptotic} gives
$\log_2P_n=c(n+1)+o(n)$, and hence $h_{\mathrm{KL}}(n)=o(n)$.

Finally, the definition is effective.  At angle $\pi/3$ the comparison point is $1/2$, and clearing
denominators gives Jacobi polynomials with coefficient bit length polynomial in $n$.  If
$x_{r,r}(a,b)$ is the largest zero of $P_r^{(a,b)}$, Nikolov's bound~\citep[Theorem~1]{Nikolov20}
gives
\[
 1-x_{r,r}(a,b)
 <
 \frac{4(a+1)(a+2)(a+4)}
 {(5a+11)\bigl(r(r+a+b+1)+(a+1)(b+1)/3\bigr)}.
\]
For the relevant parameters, $a=(d-1)/2$ and $b\le(d-1)/2$.  Setting $r=8d$ makes the
right-hand side at most
$32d^3/(160d^3)<1/2$ for $d\ge3$.  Hence the relevant Levenshtein degree is $O(d)$; the remaining
dimensions are handled directly.  Sturm sequences at $1/2$ identify the interval in
\citep[Section~5.1]{SZ24}, after which the displayed binomial formula gives $S_{n+1}$.  Certified
logarithm evaluation gives $\underline c_n$, and the binary length of $P_n$ gives
$\lceil\log_2P_n\rceil$.  Thus $h_{\mathrm{KL}}(n)$ is computable in polynomial time.

\end{document}